\documentclass[journal]{./IEEEtran}

\usepackage{psfrag}
\usepackage{cite}
\usepackage{graphicx}
\usepackage{amsmath}
\usepackage{multirow}
\usepackage{array}
\usepackage{epsfig}
\usepackage{balance}
\usepackage{enumerate}
\usepackage{graphicx}
\usepackage{times}
\usepackage{amssymb}
\usepackage{algorithm}
\usepackage{algorithmic}
\usepackage{color}

\newtheorem{lemma}{Lemma}
\newtheorem{theorem}{Theorem}
\newtheorem{definition}{Definition}
\newtheorem{property}{Property}

\newtheorem{remark}{Remark}

\begin{document}

\title{Fundamental Limits of CDF-Based Scheduling: Throughput, Fairness, and Feedback Overhead}

\author{Hu Jin, {\em Member}, {\em IEEE}, Bang Chul Jung, {\em Member}, {\em IEEE}, and
Victor C. M. Leung, {\em Fellow}, {\em IEEE}
\thanks{Manuscript received October 19, 2012; revised November 24, 2013;
accepted February 22, 2014; approved by IEEE/ACM TRANSACTIONS ON NETWORKING Editor D.
Manjunath. This paper has been published in part in the Proceedings of the IEEE
International Conference on Communications (ICC), Budapest, Hungary, June 9-13, 2013.}
\thanks{H. Jin is with the
Department of Electronics and Communication Engineering, Hanyang University, Ansan 426-791, Korea (E-mail:
hjin@hanyang.ac.kr).}
\thanks{B. C. Jung, corresponding author, is with the Department of Information and Communication Engineering \& Institute of Marine Industry, Gyeongsang National University, Tongyeong 650-160, Korea
(E-mail: bcjung@gnu.ac.kr).}
\thanks{V. C. M. Leung is with the
Department of Electrical and Computer Engineering, The University of British Columbia, Vancouver, BC, V6T 1Z4,
Canada (E-mail: vleung@ece.ubc.ca).} }

\maketitle

\begin{abstract}
In this paper, we investigate fundamental performance limits of cumulative distribution function (CDF)-based
scheduling (CS) in downlink cellular networks. CS is known as an
efficient scheduling method that can assign different time fractions for users or, equivalently, satisfy different
channel access ratio (CAR) requirements of users while exploiting multi-user diversity. We first mathematically
analyze the \textit{throughput characteristics} of CS in arbitrary fading statistics and data rate functions. It is
shown that the throughput gain of CS increases as the CAR of a user decreases or the number of users in a cell
increases. For Nakagami-$m$ fading channels, we obtain the average throughput in closed-form and investigate the
effects of the average signal-to-noise ratio, the shape parameter $m$, and the CAR on the throughput performance.
In addition, we propose a threshold-based \textit{opportunistic feedback} technique in order to reduce feedback
overhead while satisfying the CAR requirements of users. We prove that the average feedback overhead of the
proposed technique is upper bounded by $-\ln p$, where $p$ is the probability that no user satisfies the threshold
condition in a cell. Finally, we adopt a novel fairness criterion, called \textit{qualitative fairness}, which
considers not only the quantity of the allocated resources to users but also the quality of the resources. It is
observed that CS provides a better qualitative fairness than other scheduling algorithms designed for controlling
CARs of users.
\end{abstract}


\section{Introduction}
In wireless networks, independent fading of users can be exploited for multi-user diversity. In arbitrary fading
channels, the optimal user scheduling method that maximizes the sum throughput both in uplink~\cite{Knopp95} and
downlink~\cite{Tse97} is to select the user who has the largest channel gain at each time slot. Although the above
scheduling method can maximize the sum throughput, it may cause a fairness problem among users located at different
distances from the base station (BS) because the BS tends to select users who are closer to it more frequently due
to their higher average signal-to-noise ratios (SNRs).

Several approaches exist to solve the fairness problem in user scheduling. These approaches have
adopted two different criteria: throughput-based fairness~\cite{Borst01, Chung07, Viswanath02, Kim05, Liu10} and
resource-based fairness~\cite{Liu01, Sharif05, Chen06, Hwang08}. Different systems may adopt different fairness
criteria according to their design objectives. The basic idea of scheduling with throughput-based fairness is
to select the user who maximizes the system throughput, while satisfying a given throughput fairness
criterion. For example, the proportional fairness scheduler (PFS)~\cite{Viswanath02}, originally proposed in the
context of game theory~\cite{Kelly98}, maximizes the product of throughput of users. However, the long term average
throughput of PFS cannot be derived and thus we cannot determine how much resources to allocate to each user with PFS
even in stationary Rayleigh fading channels~\cite{Liu10}. Therefore, PFS does not provide a predictable system
performance. A scheduler designed to achieve throughput-based fairness in a wireless system may allow users with bad channel
conditions to occupy most resources, which degrades the throughput performance of other users.

On the other hand, with resource-based fairness the required amount of resources are assigned to each user and the
throughput obtained from the resources assigned to each user depends on the average SNR, channel statistics,
transmission techniques, etc. Hence, the user who has a higher average SNR or a better transmission technique can
achieve a higher throughput. In this paper, we focus on scheduling algorithms with resource-based fairness. The
round-robin scheduling (RRS) algorithm ~\cite{Shreedhar96} is the simplest scheduling algorithm with resource-based
fairness, which can control the assignment of time fractions for user access, referred as channel access ratios
(CARs) of users in this paper. However, RRS cannot exploit multi-user diversity in wireless communication systems.
Another scheduling method in this category is user selection based on normalized SNR (NSNR) ~\cite{Sharma05}.
Due to its analytical tractability, NSNR has been extensively investigated~\cite{Chen06, Yang06}. However, NSNR
cannot guarantee equal CARs among users when SNR distributions of users are different from each other. Liu {\em
et al.} proposed a scheduling algorithm that maximizes the sum throughput of users given their CAR
requirements~\cite{Liu01}.

Moreover, several scheduling algorithms that assign channel resources to users based on the cumulative distribution
function (CDF) values of channel gains have been proposed in independent studies including the CDF-based scheduling
(CS) algorithm~\cite{Park05}, the distribution fairness scheduling algorithm~\cite{Qin04}, and the score-based
scheduling algorithm \cite{Bonald04}. In CS~\cite{Park05}, the throughput of each user can be obtained
independently, and thus CS is robust to variations of system parameters such as traffic characteristics and number
of users in a cell~\cite{Patil09}. With these useful properties, CS has been studied under various network
scenarios such as multi-user multi-input multi-output~\cite{Kountouris08}, multi-cell coordination~\cite{Bang11},
and cheating of CDF values~\cite{Porat09}. CS was also extended to operate over heterogeneous systems where
real-time and best effort traffic coexists~\cite{Patil07}. The concept of CS was also applied to medium access
control to resolve collisions as well as exploit multi-user diversity in single-hop~\cite{Miao12} and
multi-hop~\cite{Wang06} networks. Although many studies on CS exist, the throughput characteristics of CS in
cellular downlink have not been fully investigated.

For {\em equally weighted users}, all users require the same CAR, and CS selects the user with the largest CDF value
among users in each time slot. {When all users have the same channel statistics and average SNRs, the user
selection policies of CS \cite{Park05}, Liu's scheduling algorithm~\cite{Liu01} and the distribution fairness
scheduling algorithm~\cite{Qin04} are identical.} However, for {\em unequally weighted users} who require
diverse CARs due to different service priorities, quality-of-service (QoS), or pricing policy, etc., the user
selection policies of these scheduling algorithms are quite different from each other \footnote{The score-based
scheduling algorithm proposed in \cite{Bonald04} did not consider unequally weighted users.}. Liu's scheduling
algorithm maximizes the sum throughput of all users in a cell, while the distribution fairness scheduling algorithm
maximizes the sum of the CDF values of the selected users. {However, the literature gives no indication of the
unique property of CS, which distinguishes it from Liu's scheduling algorithm and the distribution fair scheduling
algorithm under unequally weighted users.} Note that these algorithms can satisfy the CAR requirements of users but
show different throughput performance. This phenomenon motivates us to reconsider the fairness aspects of these
algorithms especially {for unequally weighted users,}  because they were originally proposed to address the
fairness issue when exploiting multi-user diversity. Satisfying CAR requirements of users is important in terms of
resource-based fairness but it is not enough to capture all aspects of fairness among users. We need another
fairness criterion to address an additional fairness aspect among users.

A primary goal of many scheduling algorithms is to exploit multi-user diversity in wireless communication systems,
and thus the degree of achieved multi-user diversity for users can be another consideration for fairness. There has
not been any suitable metric in previous studies on resource sharing, which measures the degree of achieved
multi-user diversity. In this paper, therefore, we propose a novel fairness criterion called {\em qualitative
fairness index} (QFI) to measure the degree of achieved multi-user diversity for users under the resource sharing
constraints. QFI is a positive value smaller than 1 and a scheduling algorithm is considered to be well designed
for fairly exploiting multi-user diversity if its QFI approaches the maximum value of 1. {While we show that QFI
gives a measure of the degree of achieved multi-user diversity considering unequally weighted users, we note that
other measures of the degree of achieved multi-user diversity may exist.}

In this paper, we investigate the fairness problem among users in two aspects: {\em quantitative fairness} and {\em
qualitative fairness}. Quantitative fairness stands for satisfaction on the CAR requirements, while qualitative
fairness refers to the quality of the assigned resources to users or satisfaction on the degree of achieved
multi-user diversity. A fair scheduler should satisfy both fairness criteria. It has previously been shown that
RRS, CS, Liu's scheduling algorithm and the distribution fairness scheduling algorithm all satisfy the arbitrary
CAR requirements of users, which means that they can provide quantitative fairness among users. {We further
investigate the qualitative fairness aspects of these algorithms and observe that CS shows a better performance
than other algorithms in terms of qualitative fairness. Hence, we can conclude that superior qualitative fairness
is property that distinguishes CS from other scheduling algorithms.}

In order to exploit multi-user diversity, CS requires all users to feed their CDF values back to the BS in each time
slot as in other scheduling algorithms. For practical systems, the overhead of such feedbacks is a challenging
issue especially when a large number of users exist in a cell. Therefore, it is of great interest to design a
feedback reduction scheme for CS to reduce the number of users sending feedback in each time slot. Several
threshold-based feedback reduction schemes~\cite{Gesbert04, Kim07, Hwang08} have been proposed for various
scheduling methods such as PFS and NSNR. However, none of these schemes supports different CARs among users, as CS
does. Consequently, these feedback reduction schemes cannot be applied to CS.

In this paper, we first analyze the \textit{throughput characteristics} of CS, which has not been fully
investigated in previous studies. For example, for Nakagami-$m$ fading channels, we derive the analytical
expression of the throughput and investigate the effects of the average SNR, the shape parameter $m$, and the
channel access ratio on the throughput gain.  We also propose a novel \textit{feedback reduction} scheme for CS,
which is based on a single threshold even with unequally weighted users. It is shown that the average feedback
overhead of the proposed scheme is smaller than $-\ln p$, where $p$ indicates the probability that no user
satisfies the threshold condition. Finally, we extensively investigate the \textit{fairness} aspect of CS.
Especially, we focus on the qualitative fairness of CS for a given CAR requirement. We show that CS yields a
relatively better qualitative fairness, compared with other scheduling algorithms that can control the CARs.

The rest of this paper is organized as follows:  Section \ref{sec:Proposed_scheme} introduces the system model and
reviews CS.  Section \ref{sec:Analysis} analyzes the throughput performance of CS. Section
\ref{sec:Feedback_reduction} presents the threshold-based feedback reduction scheme and analyzes its performance.
Section \ref{sec:Fairness} introduces the concept of qualitative fairness and discusses the qualitative fairness of
CS. Section \ref{sec:Numerical_results} presents the numerical results. Finally, conclusions are drawn in Section
\ref{sec:Conclusions}.

\section{System Model}\label{sec:Proposed_scheme}
We consider the downlink of a cell with a BS and $n$ users. At each time slot, the BS selects one user to receive
its transmission. The transmit power of the BS is assumed to be constant in each time slot. The BS and the users
are each assumed to have a single antenna. In time slot $t$, the received signal at the $i$-th user is given as
\begin{equation}\label{eq:Channel_model}
\mathbf{y}_i(t) = h_i(t)\mathbf{x}(t) + \mathbf{z}_i(t), ~~~~~~~i=1,2,...,n,
\end{equation}
where $\mathbf{y}(t) \in \mathbb{C}^T$ consists of $T$ received symbols, $\mathbf{x}(t) \in \mathbb{C}^T$ is the
$T$ transmitted symbols, $h_i(t) \in \mathbb{C}$ is the channel gain from the BS to the $i$-th user, and
$\mathbf{z}_i(t)\in \mathbb{C}^T$ is a zero-mean circular-symmetric Gaussian random vector~$(\mathbf{z}_i(t) \sim
\mathcal{CN}(0, \sigma^2 I_T))$. $\mathbb{C}$ denotes set of complex numbers and $I_T$ denotes the identity matrix
of size $T$. The transmit power of the BS is set to $P$, i.e., $E[|\mathbf{x}(t)|^2]/T = P$. We assume a
block-fading channel where the channel gain is constant during $T$ symbols and independently changes between time
slots. Different users may have different channel statistics. The received SNR of the $i$-th user is given by
$\gamma_i(t) = P|h_i(t)|^2/\sigma^2$. Let $F_{i}(\gamma)$ denote the CDF of the SNR of the $i$-th user, which can
be obtained from long-term observations by the user. {In each slot, the BS transmits training signals to facilitate
the users' observations on the CDF.} {If we consider $F_i(\gamma)$ as a general function of $\gamma$ and define the
corresponding output value as $U_i = F_i(\gamma)$, $U_i$ is also a random variable which is transformed from
$\gamma$ through function $F_i(\gamma)$. Then the value of $U_i$ is included in [0, 1] since any CDF value is in
[0, 1]. Moreover, the distribution of $U_i$ can be shown to be uniform in [0, 1] as follows:}
\begin{equation}\label{eq:uniform}
\begin{array}{ll}
F_{U_i}(u)  &= \Pr\{F_i(\gamma) < u\} \\
            &= \Pr\{\gamma < F_i^{-1}{(u)}\} \\
            &= F_i( F_i^{-1}{(u)}) \\
            &= u.
\end{array}
\end{equation}
In this paper, we assume that all users' channels are stationary and the channel statistics of each user are
assumed to be independent from those of other users. {Equation (\ref{eq:uniform}) indicates that all users' CDF
values have the same uniform distribution while their CDFs may not be identical\footnote{In this paper, the {\em
CDF} indicates the function $F_i(\gamma)$ itself while the {\em CDF value} indicates the output value of the CDF
with a specific input of $\gamma$.}. CS exploits this property in fair multi-user scheduling. }

Let $w_i (>0)$ denote the weight of the $i$-th user. The weight indicates the user's CAR compared to other users,
which means that the ratio between the $i$-th and $j$-th users' channel access opportunities is given by $w_i
/w_j$. If there are $n$ users in the system, the $i$-th user's CAR is $\alpha_i = \frac{w_i}{\sum_{j=1}^n w_j}.$
With CS, the feedback information of the $i$-th user is $[F_{i}(\gamma_i(t))]^{\frac{1}{w_i}}$ at time slot $t$ and
the index of the user selected at the BS is given by
\begin{equation}\label{eq:CDF_policy}
\arg \max_{i \in \{1,2,..., n\}}[F_{i}(\gamma_i(t))]^{\frac{1}{w_i}}.
\end{equation}
It has been shown in \cite{Park05} that this scheduling algorithm yields a CAR of $\alpha_i$ for the $i$-th user.
{Note that with CS, the users are in charge of observing the CDFs through long-term observations and calculating
the CDF values for feedback in each time slot while the BS only compares the information sent from the users in
each time slot and does not need to know all users� CDFs.}

\section{Throughput Characteristics of CS}\label{sec:Analysis}

In order to investigate the throughput performance of CS, we first analyze the SNR distribution of the selected
user. We start with the following lemma.
\begin{lemma}\label{lm:CDF}
Let $F(\gamma)$ be the SNR distribution of a user and $\alpha (\in [0, 1])$ be the CAR of the user. With CS, the
SNR distribution of the user given it is selected is expressed as:
\begin{equation}\label{eq:CDF}
F_{\rm Sel}^{\rm CS}(\gamma) = [F(\gamma)]^{\frac{1}{\alpha}}.
\end{equation}
\end{lemma}

\begin{proof}
Let $(w_1, w_2, ..., w_n)$ be the weight vector of $n$ users. The $i$-th user feeds back
the value of $[F_{i,t}(\gamma_i(t))]^{\frac{1}{w_i}}$ to the BS. Then, the feedback
information received at the BS in each time slot is given by $(U_1^{\frac{1}{w_1}},
U_2^{\frac{1}{w_2}}, ..., U_n^{\frac{1}{w_n}})$ and the BS selects the user with the
largest value of the feedback information. Then, the probability that the $i$-th user is
selected is given as:
\begin{equation}
\begin{array}{ll}
\Pr\{U_j^{\frac{1}{w_j}} < U_i^{\frac{1}{w_i}} , \forall j \in \{1, 2, \cdots, i-1, i+1, \cdots, n\}\} \\
=\int_0^1 \prod_{j=1, j \neq i}^n \Pr\{U_j^{\frac{1}{w_j}} < u^{\frac{1}{w_i}}\} f_{U_i}(u)du\\
=\int_0^1 \prod_{j=1, j \neq i}^n \Pr\{U_j < u^{\frac{w_j}{w_i}}\} f_{U_i}(u)du\\
=\int_0^1\prod_{j=1, j \neq i}^n u^{\frac{w_j}{w_i}} du\\
=\int_0^1  u^{\sum_{j=1, j \neq i}^n\frac{w_j}{w_i}} du\\
=\displaystyle \frac{w_i}{\sum_{j=1}^n w_j} \\
=\alpha_i.
\end{array}
\end{equation}
Thus, the CAR of the $i$-th user is given by $\alpha_i={w_i}/{\sum_{j=1}^n w_j}$. If all users
have identical weights, the CAR of each user is equal to ${1}/{n}$. The SNR distribution of the
$i$-th user given it is selected is expressed as:
\begin{equation}\label{eq:selected_CDF}
\begin{array}{ll}
F_{i, {\rm Sel}}(\gamma)  &=\Pr\{ \gamma_i \leq \gamma | \text{the $i$-th user is selected}\} \\
                    &= \frac{\Pr\{ \gamma_i \leq \gamma, \text{the $i$-th user is selected}\}}{\Pr\{ \text{the $i$-th user is selected}\}}\\
                    &= \frac{\Pr\{ \gamma_i \leq \gamma, U_j^{\frac{1}{w_j}} < U_i^{\frac{1}{w_i}}, \forall j \in \{1, 2, \cdots, i-1, i+1, \cdots, n
                    \}\}}{\alpha_i}\\
                    &= \frac{\Pr\{ \gamma_i \leq \gamma, U_j < [F_i(\gamma_i)]^{\frac{w_j}{w_i}}, \forall j \in \{1, 2, \cdots, i-1, i+1, \cdots, n  \}\}}{\alpha_i}\\
                    & = \frac{\int_0^\gamma [F_i(\gamma_i)]^{\sum_{j=1,j\neq i}^n\frac{w_j}{w_i}} d
                    F_i(\gamma_i)}{\alpha_i}\\
                    &= \frac{\frac{w_i}{\sum_{j=1}^n w_j}[F_i(\gamma)]^{\frac{\sum_{j=1}^n w_j}{w_i}}}{\alpha_i}\\
                    &= [F_i(\gamma)]^{\frac{1}{\alpha_i}}.
\end{array}
\end{equation}
\end{proof}

To investigate the throughput behavior of CS, we first define the following function:
\begin{definition} \textit{(Universal Throughput Function)} \label{def:universal_thru_fn}
\begin{equation}\label{eq:Universal_throughput_extended}
S(x, \alpha) = \int_{F^{-1}(x)}^\infty R(\gamma)d[F(\gamma)]^{\frac{1}{\alpha}},
\end{equation}
where $x$ and $\alpha$ are values taken in $[0, 1]$, $F(\gamma)$ is the SNR
distribution, which is an increasing function of $\gamma$, $F^{-1}(x)$ is the inverse
function of $F(\gamma)$, and $R(\gamma)$ is the data rate function corresponding to the
instantaneous SNR value.
\end{definition}

We assume $R(\gamma)$ is an increasing function of $\gamma$ since a higher SNR enables a higher data rate in
general. If $S(x, \alpha) < \infty$, we can obtain the following properties and we will use them to find the
interesting throughput behavior of CS later.
\begin{property}\label{pr:alpha_inc}
$\alpha S(x, \alpha)$ is an increasing function of $\alpha$.
\end{property}
\begin{property}\label{pr:alpha_dec}
$S(x^\alpha, \alpha)$ is a decreasing function of $\alpha$.
\end{property}
\begin{property}\label{pr:p_dec}
$S(x, \alpha)$ is a decreasing function of $x$.
\end{property}
\begin{property}\label{pr:p_inc}
$\frac{S(x, \alpha)}{1 - x^\frac{1}{\alpha}}$ is an increasing function of $x$.
\end{property}
\begin{proof}
See Appendix~\ref{apx:Property}.
\end{proof}

For a given CAR of a user~$\alpha$, we define {\em throughput gain} $(g_{\rm CS})$ as the throughput ratio of CS to
RR. Based on \textit{Properties \ref{pr:alpha_inc}} and \textit{\ref{pr:alpha_dec}}, we obtain the following
result.
\begin{theorem}\label{tm:CDF_th}
{With CS, the throughput of a user experiencing arbitrary stationary fading channel increases as its CAR increases,
but the throughput gain decreases as the CAR increases.}
\begin{proof}
Let $\alpha$ and $F(\gamma)$ be the CAR and the SNR distribution of a user, respectively. The throughput of the
user with CS is given as:
\begin{equation}\label{eq:S_CDF}
S_{\rm CS}(\alpha) =\alpha \int_{0}^\infty R(\gamma)d[F(\gamma)]^{\frac{1}{\alpha}} = \alpha S(0, \alpha).
\end{equation}
Applying \textit{Property \ref{pr:alpha_inc}}  with $x=0$, we can observe that $S_{\rm CS}$ is an increasing
function of $\alpha$. In other words, the throughput decreases as the CAR decreases. In practice, the CAR decreases
as the number of users increases in a cell.

Similarly, the throughput of the user with RRS is given as:
\begin{equation}
S_{\rm RRS}(\alpha) = \alpha \int_{0}^\infty R(\gamma)dF(\gamma) = \alpha S(0, 1) = \alpha E[R],
\end{equation}
where $E[R]$ indicates the average data rate for the user, defined as $S(0, 1)$ from
(\ref{def:universal_thru_fn}). The throughput gain is expressed as:
\begin{equation} \label{eq:CDF_th_gain}
g_{\rm CS}(\alpha) = \displaystyle \frac{S_{\rm CS}(\alpha)}{S_{\rm RRS}(\alpha)}= \displaystyle  \frac{S(0,
\alpha)}{S(0, 1)} =\displaystyle  \frac{S(0, \alpha)}{E[R]}\geq 1.
\end{equation}
From \textit{Property \ref{pr:alpha_dec}}  with $x=0$, the throughput gain increases as the
CAR decreases.
\end{proof}
\end {theorem}

Furthermore, if we apply \textit{Property \ref{pr:alpha_dec}} with $x=0$, then we can also observe that $S_{\rm
CS}$ is larger than $S_{\rm RRS}$, which means that CS always provides a higher throughput than RRS; i.e., the
throughput gain is always larger than $1$. Although CS provides a better throughput performance than RRS, it does
not guarantee the optimal throughput for a given CAR. For example, we consider the cellular downlink where there
exist two users in the cell and the channel of the first user is Rayleigh distributed, while the channel of the
second user is constant. Both users are assumed to have the same CAR, which is equal to $1/2$. Since the achievable
data rate of the second user is constant at any time, the throughput of the first user is maximized when the BS
selects it if $F_1(\gamma_1) > 1/2$ where $\gamma_1$ denotes the SNR of the first user. For a given CAR, the
optimal throughput of a user is obtained by the following lemma\footnote{A similar upper-bound was also given in
\cite{Park05}, but a rigorous proof was not provided.}:
\begin{lemma}\label{tm:UB}
Let $F(\gamma)$ and $R(\gamma)$ be the SNR distribution and the data rate function of a user, respectively. For a
given CAR requirement $\alpha$, the throughput of the user is upper bounded by
\begin{equation}\label{eq:S_UB}
S_{\rm UB}(\alpha) = S(1 - \alpha, 1).
\end{equation}
\end{lemma}
\begin{proof}
{See Appendix~\ref{apx:lm_UB}.}
\end{proof}

The above lemma implies that, for a given CAR of $\alpha$, the optimal scheduling algorithm in terms of throughput
is to select the user with SNR such that $F(\gamma) \geq 1-\alpha$. Based on \textit{Properties \ref{pr:p_dec}} and
\textit{\ref{pr:p_inc}} with $x=1-\alpha$, we observe that the throughput upper-bound decreases as the CAR
decreases. We define another throughput gain as:
\begin{equation} \label{eq:UB_th_gain}
g_{\rm UB}(\alpha)=\frac{S_{\rm UB}(\alpha)}{S_{\rm RRS}(\alpha)}.
\end{equation}
The throughput gain of the upper-bound in (\ref{eq:UB_th_gain}) increases as the CAR decreases. The following
theorem states the throughput relationship between CS and the optimal scheduling algorithm:
\begin{theorem}\label{tm:UB_achivability}
{If the supported data rate of a user has a maximum value, the throughput of CS approaches the throughput upper-bound
as the CAR decreases to zero.} In other words, for a given condition that $R(\gamma) = R_{\rm th}$ for $\gamma >
\gamma_{0}$, we obtain
\begin{equation}
\lim_{\alpha \rightarrow 0}\frac{S_{\rm CS}(\alpha)}{S_{\rm UB}(\alpha)} =\lim_{\alpha \rightarrow 0} \frac{\alpha
S(0, \alpha)}{S(1 -\alpha, 1)} = 1.
\end{equation}
\end{theorem}
\begin{proof}
See Appendix~\ref{apx:lm_UB_achivability}.
\end{proof}

In practice, the supported number of levels of the modulation and coding scheme (MCS) is finite and the maximum
data rate is limited. Hence, CS can achieve the throughput upper-bound as the CAR tends to zero. If the data rate
function has no upper limit, then the throughput of a user with  CS yields the following behavior:
\begin{theorem}\label{tm:low_up}
For a given CAR requirement $\alpha$ and no upper limit for the data rate function, the throughput of CS is upper-bounded by $S_{\rm UB}(\alpha)$ and is lower-bounded by:
\begin{equation}\label{eq:low_up}
S_{\rm CS}(\alpha) \geq \alpha R(F^{-1}(1 + \alpha\ln\alpha))[1 - (1 + \alpha\ln\alpha)^{\frac{1}{\alpha}}].
\end{equation}
Furthermore, if $\lim_{\alpha \rightarrow 0} S_{\rm UB}(\alpha) = 0$, the throughput gain of CS has the following
characteristics:
\begin{align}
\displaystyle \lim_{\alpha \rightarrow 0} g_{\rm CS}(\alpha)& \geq g_{\rm LB} = \lim_{\alpha \rightarrow 0}\frac{ R(F^{-1}(1+\alpha\ln(\alpha)))}{E[R]}, \\
\displaystyle \lim_{\alpha \rightarrow 0} g_{\rm CS}(\alpha)& \leq g_{\rm UB} = \lim_{\alpha \rightarrow 0}\frac{
R(F^{-1}(1-\alpha))}{E[R]}.
\end{align}
If $\lim_{\gamma \rightarrow \infty} R(\gamma) = \infty$, then the throughput gain $g_{\rm CS}(\alpha)$ tends to
infinity as $\alpha$ decreases to zero.
\end{theorem}
\begin{proof}
See Appendix~\ref{apx:lm_low_up}.
\end{proof}

According to \textit{Theorem~\ref{tm:low_up}}, an SNR distribution with large values of $F^{-1}(1+\alpha\ln\alpha)$
and $F^{-1}(1-\alpha)$ yields a high throughput gain with CS when the CAR is small enough. Fig.~\ref{fig:CDF} shows
the CDF of the received SNR of a user in Nakagami-$m$ fading channel when the average SNR is set to $0$ dB. When
$\alpha$ is small enough, for the example where $\alpha = 0.1$ in the figure, both
$F^{-1}(1+\alpha\ln\alpha)=F^{-1}(0.77)$ and $F^{-1}(1-\alpha)=F^{-1}(0.9)$ become smaller as the shape parameter
$m$ increases, which means that the throughput gain of CS decreases as the shape parameter increases. On the other
hand, the outage probability, which is the performance metric of interest in many wireless systems, decreases as
the shape parameter increases. If we set the SNR threshold for the outage to 0.5 in the figure, the outage
probability decreases from 0.39 to 0.02 as the shape parameter $m$ increases from 1 to 10. Thus, there is a
tradeoff between low outage probability and high throughput gain.

\begin{figure}
\centering
\includegraphics[width=3.5in]{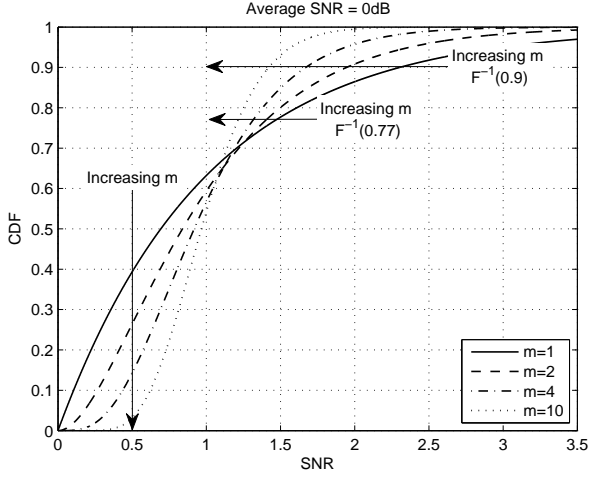}
    \caption{CDF curves for Nakagami-$m$ fading channels.}
    \label{fig:CDF}
\end{figure}

For a representative example, in the rest of this section, we analyze the throughput of a user with CS in
Nakagami-$m$ fading channels. We assume the data rate function as $R(\gamma) = \log_2(1+\gamma)$, which is the
Shannon capacity. In Nakagami-$m$ fading channels, the SNR distribution of a user follows the Gamma distribution
whose probability density function (pdf) is given as:
\begin{equation}
f_{m, \overline \gamma}(\gamma) = \frac{1}{\Gamma(m)}\left(\frac{m}{\overline
\gamma}\right)^{m} \gamma^{m-1}e^{-\frac{m}{\overline \gamma}\gamma},
\end{equation}
where $m$ denotes the shape parameter, $\overline \gamma$ denotes the average SNR,  and $\Gamma(m)$ indicates the
Gamma function defined as $\int_0^\infty e^{-y}y^{m-1}dy$. If $m$ is a positive integer, the corresponding CDF is
expressed as:
\begin{equation}
F_{m, \overline \gamma}(\gamma) = 1 -\sum_{j=0}^{m-1}\frac{1}{j!}\left(\frac{m}{\overline
\gamma}\right)^{j} \gamma^je^{-\frac{m}{\overline \gamma}{\gamma}}.
\end{equation}
If $K=\frac{1}{\alpha}$ and $K$ is an integer, by extending the analysis in
\cite{Chen06}, the universal throughput function, $S(x, \alpha)$, can be expressed as:
\begin{equation}\label{eq:S_nakagami}
\begin{array}{ll}
S(x, \frac{1}{K}) &= \log_2(1+\gamma_{\rm th}) \{1 - [F_{m, \overline \gamma}(\gamma_{\rm th})]^K \} \\
                  &~~~+ \log_2(e)\sum_{k=1}^{K}   \sum_{j=0}^{k(m-1)} (-1)^{k+1} \cdot \\
                  &~~~{K \choose k} \underline{c(j, k)} \left(\frac{m}{\overline \gamma}\right)^j \underline
{T(\gamma_{\rm th}, j, \frac{\overline \gamma}{km})},
\end{array}
\end{equation}
where the term $\gamma_{\rm th}$ is the SNR satisfying $F_{m, \overline \gamma}(\gamma_{\rm th})=x$, the term
$c(j, k)$ is defined as:
\begin{equation}
\begin{array}{l}
c(0, k) = 1, ~~~~~~\\
c(1, k) = k,\\
c(k(m-1), k) = [(m-1)!]^{-k},\\
c(j, k) = \frac{1}{j}\sum\limits_{l=1}^{\min(j, m-1)}\frac{l(k+1)-j}{l!}c({j-l}, k),\\
~~~~~~~~~~~~~~~~~~~~~~~~~~~{\rm for~}2\leq j \leq k(m-1),
\end{array}
\end{equation}
and the term $T(\gamma_{\rm th}, j, \theta)$ is defined as:
\begin{equation}
\begin{array}{l}
T(\gamma_{\rm th}, j, \theta) = e^{\frac{1}{\theta}} \left\{ (-1)^{j} E_1\left({\frac{1 + \gamma_{\rm th}}{\theta}}\right) \right. \\
\left.\!\!+ \!\!\sum\limits_{i=1}^j {j \choose i} (-1)^{j-i} \theta^i (i-1)!\left[\sum\limits_{l=0}^{i-1}\frac{1}{l!}\left(\frac{1 + \gamma_{\rm th}}{\theta} \right)^le^{-\frac{1 + \gamma_{\rm th}}{\theta} }\right] \right\},\\
\end{array}
\end{equation}
where the exponential integral function of the first kind is defined as $ E_1(y) = \int_y^\infty
\frac{e^{-t}}{t}dt$. Thus, based on (\ref{eq:S_nakagami}), the values of $S_{\rm RRS}(\alpha)$, $S_{\rm
CS(\alpha)}$, and $S_{\rm UB}(\alpha)$ can also be obtained.

It is well known that a larger number of users result in a higher multi-user diversity.
Conventionally, this phenomenon has been observed by investigating the increasing scale of
the sum throughput when the number of users increases to infinity in cellular
systems~\cite{Viswanath02, Jung11, Jung12} and cognitive networks~\cite{Ban09, Tajer10}. The increasing scale
represents how fast the throughput increases as the number of users in a network increases. Since we
consider resource-based fairness in this paper, we alternatively investigate the
increasing scale of the throughput gain of each user when the number of users
increases to infinity or, equivalently, the CAR decreases to zero.
The inverse function of $F_{m, \overline \gamma}(\gamma)$ is given as \cite{Sharif05}:
\begin{equation}
\begin{array}{l}
F_{m, \overline \gamma}^{-1}(1-\alpha) =\\~~~~ \frac{\overline \gamma}{m}
\left[\ln(\frac{1}{\alpha}) +
(m-1)\ln\ln(\frac{1}{\alpha})+O(\ln\ln\ln(\frac{1}{\alpha}))\right].
\end{array}
\end{equation}
When $\alpha$ approaches to $0$, the upper- and lower-bound of the throughput gain with CS is given as
\begin{equation}\label{eq:UB_scale_nakagami}
\begin{array}{ll}
g_{\rm UB}= \lim\limits_{\alpha \rightarrow 0}\frac{ R(F_{m, \overline
\gamma}^{-1}(1-\alpha))}{E[R]}
\\=\log_2e\lim\limits_{\alpha \rightarrow 0}\frac{ \ln\left(1+\frac{\overline \gamma}{m}
\left[\ln(\frac{1}{\alpha}) +
(m-1)\ln\ln(\frac{1}{\alpha})+O(\ln\ln\ln(\frac{1}{\alpha}))\right]\right)}{E[R]}\\
\\=\log_2e\lim\limits_{\alpha \rightarrow 0}\frac{ \ln(\frac{\overline \gamma}{m}) +
\ln\ln(\frac{1}{\alpha})+O(\ln\ln\ln(\frac{1}{\alpha}))}{E[R]},\\
\end{array}
\end{equation}
\begin{equation}\label{eq:LB_scale_nakagami}
\begin{array}{ll}
g_{\rm LB}= \lim\limits_{\alpha \rightarrow 0}\frac{ R(F_{m, \overline
\gamma}^{-1}(1+\alpha\ln\alpha))}{E[R]}
\\=\log_2e\cdot  \\~~ \lim\limits_{\alpha \rightarrow 0}\frac{ \ln\left(1+\frac{\overline \gamma}{m}
\left[\ln(\frac{-1}{\alpha\ln\alpha}) +
(m-1)\ln\ln(\frac{-1}{\alpha\ln\alpha})+O(\ln\ln\ln(\frac{-1}{\alpha\ln\alpha}))\right]\right)}{E[R]}\\
= \log_2e\lim\limits_{\alpha \rightarrow 0}\frac{1}{E[R]} \ln\left\{1+\frac{\overline \gamma}{m} \cdot \left[(\ln(\frac{1}{\alpha}) - \ln\ln(\frac{1}{\alpha})) \right.\right.\\
 ~~~~~~~~~~~~~~~~~~+ (m-1)\ln[\ln(\frac{1}{\alpha}) - \ln\ln(\frac{1}{\alpha})]\\
 ~~~~~~~~~~~~~~~~~\left. \left.+O(\ln\ln[\ln(\frac{1}{\alpha}) - \ln\ln(\frac{1}{\alpha})])\right]\right\}\\
=\log_2e\lim\limits_{\alpha \rightarrow 0}\frac{ \ln(\frac{\overline \gamma}{m}) +
\ln\ln(\frac{1}{\alpha})+O(\ln\ln\ln(\frac{1}{\alpha}))}{E[R]},
\end{array}
\end{equation}
respectively. Since $\ln(\frac{\overline \gamma}{m}) \ll \ln\ln(\frac{1}{\alpha})$ as $\alpha \rightarrow 0$, both
$g_{\rm UB}$ and $g_{\rm LB}$ increase in a scale of $\frac{\log_2e\cdot\ln\ln(\frac{1}{\alpha})}{E[R]}$ in
Nakagami-$m$ fading channels. Therefore, extending \textit{Theorem~\ref{tm:low_up}}, we have the following remark:
\begin{remark}\label{lm:scale_nakagami}
In Nakagami-$m$ fading channels, the throughput gain of CS increases with the optimal scale of
$\frac{\log_2e\cdot\ln\ln(\frac{1}{\alpha})}{{E[R]}}$ as $\alpha$ decreases to zero. With equally weighted users,
the increasing scale is given by $\frac{\log_2e\cdot\ln\ln(n)}{E[R]}$ as $n$ increases to infinity, where $n$
indicates the number of users in a cell.
\end{remark}

If the term $\ln\ln(\frac{1}{\alpha})$ is not large enough, the value $\ln(\frac{\overline \gamma}{m})$ also
affects the throughput gain as shown in (\ref{eq:UB_scale_nakagami}) and (\ref{eq:LB_scale_nakagami}), i.e., the
effect of the average SNR $\overline \gamma$ and the shape parameter $m$ is not negligible. In this case, we
have the following remark:
\begin{remark}\label{lm:scale_nakagami_2}
If $\overline \gamma \gg m$, we have $g_{\rm UB}\approx g_{\rm LB} \approx \log_2e \cdot\frac{ \ln({\overline
\gamma}) + \ln\ln(\frac{1}{\alpha})}{E[R]}$ since the effect of $m$ on the throughput gain is negligible. Moreover,
if $\overline \gamma \rightarrow \infty$ and $\ln \overline \gamma \gg \ln\ln(\frac{1}{\alpha})$, we have $g_{\rm
UB}\approx g_{\rm LB} \approx \frac{ \log_2({\overline \gamma})}{E[R]}\approx 1$. Thus, no throughput gain is
expected in the high SNR regime. On the other hand, if $\overline \gamma$ is not large enough, the shape parameter
$m$ may affect the throughput gain. A larger shape parameter reduces the throughput gain of CS as shown in
(\ref{eq:UB_scale_nakagami}) and (\ref{eq:LB_scale_nakagami}).
\end{remark}

\section{Feedback Reduction for CS}\label{sec:Feedback_reduction}
CS requires all users to feed their CDF values back to the BS at each time slot, which may cause severe feedback
overhead as the number of users in a cell increases. Several threshold-based feedback reduction
schemes~\cite{Gesbert04, Kim07, Hwang08} have been proposed for various scheduling algorithms such as PFS and NSNR.
However, none of these schemes supports different CARs among users, as CS does. Consequently, these feedback
reduction schemes cannot be applied to CS. In this section, we propose CS-FR, a novel feedback reduction scheme for
CS, to reduce the feedback overhead. To the best of our knowledge, CS-FR is the first feedback reduction scheme
that considers diverse users who require different CARs in scheduling.

\subsection{Threshold Design and Channel Access Ratio}
For {\em equally weighted users} in a cell, since all users send feedback information that is identically and
uniformly distributed between $[0, 1]$, we can simply set the same threshold $\eta_{\rm th}$ for all users to
achieve the identical CAR. If the feedback information of the $i$-th user, $U_i$, is larger than $\eta_{\rm th}$,
the $i$-th user sends $U_i$ to BS. If no user satisfies the condition, the BS does not receive any feedback
information from the users and it selects a user in a RRS manner. When no feedback happens in the slot, we call
such a slot a no-feedback (NFB) slot. We further define a slot in which more than one users send feedback to the BS as
a feedback (FB) slot. For {\em unequally weighted users}, the difficulty in determining the thresholds is to
satisfy the CARs in both FB and NFB slots. Intuition tells us that different users may have different threshold
values due to their different weights. However, we show in the following theorem that it is possible to satisfy the
CARs of different users with the same threshold $\eta_{\rm th}$ for all users.
\begin{theorem} \label{tm:acces_ratio_OF}
The CARs of the users with CS-FR are still maintained if the threshold of all user is set to $p^{1/\sum_{j=1}^n
w_j}$, where $p$ and $w_j$ denote the NFB probability and the weight of the $j$-th user, respectively.
\end{theorem}
\begin{proof}
Given the threshold $\eta_{\rm th}$ for all users, the $i$-th user feeds back the value
$U_i^{\frac{1}{w_i}}$ if it is larger than $\eta_{\rm th}$. With this setting, we show that
the CAR of the $i$-th user in the NF slots is equal to $\alpha_i =
\frac{w_i}{\sum_{j=1}^nw_j}$.

With the proposed threshold setting for CS-FR, the NFB probability is given by:
\begin{equation}\label{eq:P_out}
\begin{array}{ll}
 p  &= \Pr\left\{U_j^{\frac{1}{w_j}} < \eta_{\rm th},\forall j \in \{1, 2, \cdots, n\}\right\}\\
              &= \prod_{j=1}^n \eta_{\rm th}^{w_j}\\
              &= \eta_{\rm th}^{\sum_{j=1}^n w_j}.
\end{array}
\end{equation}
For a given NFB constraint $p$, the threshold $\eta_{\rm th}$ can be set to
$p^{1/\sum_{j=1}^n w_j}$. Hence, the selection probability for the $i$-th user in each FB
slot is
\begin{equation}
\begin{array}{l}
\text{Pr}\{{\rm user~}i{\rm~is~selected} | {\rm FB~slot}\} \\
=\frac{\text{Pr}\{{\rm user~}i{\rm~is~selected,~the~slot~is~FB~slot} \}}{\text{Pr}\{ \text{the slot is FB slot}\}}\\
=\frac{\text{Pr}\{ U_i^{\frac{1}{w_i}} > \eta_{\rm th} ~\&~U_i^{\frac{1}{w_i}} > U_j^{\frac{1}{w_j}}, \forall j \in \{1, 2, \cdots,i-1, i+1, \cdots n\}  \}}{1 - \text{Pr}\{U_j^{\frac{1}{w_j}} < \eta_{\rm th}, \forall j \in \{1, 2, \cdots,n\}\}}\\
=\frac{\int_{\eta_{\rm th}^{w_i}}^1 \prod_{j=1, j \neq i}^n \text{Pr}\{U_j < u^{\frac{w_j}{w_i}}\} f_{U_i}(u)du}{1 - p}\\
=\frac{\int_{\eta_{\rm th}^{w_i}}^1 u^{\sum_{j=1, j \neq i}^n\frac{w_j}{w_i}} du}{1 - p}
\\=\frac{\frac{w_i}{\sum_{j=1}^n w_j} \left( 1 - \eta_{\rm th}^{\sum_{j=1}^n
w_j}\right)}{1 - p} \\
= \frac{w_i}{\sum_{j=1}^n w_j} \\
= \alpha_i.
\end{array}
\end{equation}

In the NFB slots, the users are selected with RRS~(or random scheduling) so that the CAR $\alpha_i$ for the $i$-th
user is still maintained. Thus, the total CAR for the $i$-th user is
\begin{equation}
\alpha_i\Pr\{\text{FB slot}\} + \alpha_i\Pr\{\text{NFB slot}\} = \alpha_i.
\end{equation}
\end{proof}

Note that we do not assume any specific channel distribution in \textit{Theorem~\ref{tm:acces_ratio_OF}} and it can
be applied to any channel distribution. Notably, selecting the same threshold value for all users who have
different CARs substantially simplifies the system design and implementation.

\subsection{Feedback Overhead}\label{subsec:feedback_overhead}
In this subsection, we consider the average feedback overhead with CS-FR.
\begin{theorem}\label{tm:AVG_FB}
With CS-FR, the average feedback overhead in each slot is upper-bounded by $n \left(1 - p^{\frac{1}{n}}\right)$,
where $p$ denotes the NFB probability. The equality holds when all users are equally weighted. Another upper-bound
of the feedback overhead is given by $-\ln p$, which is valid regardless of the number of users and the weights of
users.
\end{theorem}
\begin{proof}
For the $i$-th user, the average feedback overhead in each slot is given as:
\begin{equation}
\begin{array}{ll}
\mu_i = \Pr\{U_i^{\frac{1}{w_i}} \geq \eta_{\rm th}\} = \Pr\{U_i \geq \eta_{\rm
th}^{w_i}\} = 1 - \eta_{\rm th}^{w_i} \\~~~~~~~= 1 - p^{\frac{w_i}{\sum_{j=1}^n w_j}} = 1
- p^{\alpha_i}.
\end{array}
\end{equation}
The average feedback overhead in each slot in a cell is given as:
\begin{equation}
\mu = \sum_{i=1}^n \mu_i = n \left(1 - \frac{1}{n}\sum_{i=1}^n p^{\alpha_i}\right).
\end{equation}
Since $f(x)=p^x$ is a convex function of $x$ in a region $0<p<1$, we have
\begin{equation}
\mu \leq n \left(1 - p^{\frac{1}{n}\sum_{i=1}^n \alpha_i}\right) = n \left(1 -
p^{\frac{1}{n}}\right).
\end{equation}
The equality holds when $\alpha_1=\alpha_2=...=\alpha_n$, i.e., all users have the same
weight. Using the fact that $x(1-p^{\frac{1}{x}})$ is an increasing function over $x$ for
$x>0$ and $0<p<1$, and $\lim_{n \rightarrow \infty}(1-\frac{x}{n})^n = e^{-x}$, we have
\begin{equation}
\mu \leq \lim_{n \rightarrow \infty}n \left(1 - p^{\frac{1}{n}}\right) = -\ln p.
\end{equation}
\end{proof}

\subsection{Throughput analysis}
In order to analyze the throughput characteristic of CS-FR, we first investigate the SNR distribution for a user
given that it is selected.
\begin{theorem}\label{lm:CDF_FR}
With CS-FR, if a user's SNR distribution is $F(\gamma)$, its CAR is $\alpha \in [0, 1]$, and the NFB probability is
$p$, the SNR distribution given this user is selected is obtained as
\begin{equation}\label{eq:CDF_FR}
F_{{\rm Sel}}(\gamma)  = \left\{
                        \begin{array}{ll}
                        p^{(1 - \alpha)} F(\gamma),& {\rm if~}0<\gamma < F^{-1}(p^{ \alpha}),\\
                        F(\gamma)^{\frac{1}{\alpha}}, &{\rm if~}\gamma \geq F^{-1}(p^{\alpha}).
                        \end{array}
                        \right.
\end{equation}
\end{theorem}
\begin{proof}
See Appendix~\ref{apx:F_i_sel}.
\end{proof}

We also define the following throughput function for analyzing the throughput of CS-FR.
\begin{definition}
\begin{equation}
S_L(x, \alpha)  = \int_0^{F^{-1}(x)} R(\gamma)d[F(\gamma)]^{\frac{1}{\alpha}} = S(0,
\alpha) - S(x, \alpha).\end{equation}
\end{definition} Then, $S_L(x, \alpha)$ and $S(x, \alpha,)$ have the following properties:
\begin{property}\label{pr:SL_alpha_inc}
$\alpha S_L(x, \alpha)$ is an increasing function of $\alpha$.
\end{property}
\begin{property}\label{pr:FR_p_dec}
$S(x^{\alpha}, \alpha) + x^{1-\alpha}S_L(x^{\alpha}, 1)$ is a decreasing function of $x$.
\end{property}
\begin{proof}
See Appendix~\ref{apx:Property}.
\end{proof}
Based on (\ref{eq:CDF_FR}), the throughput of CS-FR is calculated as
\begin{equation}\label{eq:S_CDFOF}
\begin{array}{ll}
S_{\rm CS-FR}(\alpha, p) &= \alpha \int_{F^{-1}(p^{\alpha})}^\infty R(\gamma) d [F(\gamma)]^{\frac{1}{\alpha}} \\
&~~~+ \alpha p^{1-\alpha} \int_0^{F^{-1}(p^{\alpha})} R(\gamma) d F(\gamma) \\
 &= \alpha S(p^{\alpha}, \alpha) + \alpha  p^{1-\alpha}S_L(p^{\alpha}, 1).
\end{array}
\end{equation}
We can observe that the throughput of any user depends on its CAR $\alpha$ and the NFB probability $p$ and is
independent from other users. From {\em Property \ref{pr:FR_p_dec}}, we can conclude that $S_{\rm CS-FR}$ is an
increasing function of $p$. Hence, there is no optimal threshold for CS-FR and, in order to obtain a higher
throughput, we should reduce the value of $p$. When $p=0$, CS-FR is identical to CS while CS-FR is identical to RRS
when $p=1$. Thus, CS-FR always shows a better throughput performance than RRS and a worse throughput performance
than CS. Compared to CS, the lower- and upper-bound throughputs of CS-FR are characterized by the following theorem:
\begin{theorem}\label{lm:FR_Low}
The lower- and upper-bound of $S_{\rm CS-FR}(\alpha, p)$ are given as
\begin{equation} \label{eq:FR_Low}
1-p  \leq  1-p + \alpha p^{2-\alpha } \leq \frac{S_{\rm CS-FR}(\alpha, p)}{S_{\rm CS}(\alpha)} \leq 1,
\end{equation}
\end{theorem}
\begin{proof}
For the proof of the lower-bound, see Appendix~\ref{apx:lm_FR_Low}, which applies {\em Properties \ref{pr:p_inc}}
and {\em \ref{pr:SL_alpha_inc}}. The upper-bound can be obtained from {\em Property~\ref{pr:FR_p_dec}} where the
case of $x=0$ stands for $S_{\rm CS}(\alpha)$.
\end{proof}

From the lower bound, we can conclude that the throughput loss ratio of CS-FR to CS is smaller than the NFB
probability $p$. Note that {\em Theorem \ref{lm:FR_Low}} is applicable to any data rate function and channel
statistics. {\em Theorem \ref{tm:AVG_FB}} and {\em Theorem \ref{lm:FR_Low}} lead to the following remarks for
CS-FR:
\begin{remark}\label{rm:CS-FR}
1) There is a tradeoff between throughput and feedback overhead. A larger feedback overhead gives a higher
throughput because they are both decreasing functions of $p$. 2) The feedback overhead is upper-bounded by the
negative natural logarithm of the throughput loss ratio, i.e., if each user can tolerate the throughput loss of at
most $p$ compared to CS, we can design CS-FR with the average feedback overhead smaller than $-\ln p$.
\end{remark}

In the case of Nakagami-$m$ fading channels, we can apply (\ref{eq:S_nakagami}) to derive the throughput
performance of CS-FR.

\section{Fairness Aspect of CS}\label{sec:Fairness}

{Although CS satisfies the CAR requirements and has interesting properties as discussed in Sections
\ref{sec:Analysis} and \ref{sec:Feedback_reduction}, the specific property of CS that distinguishes it from Liu's
scheduling algorithm and the distribution fairness scheduling algorithm, both of which also satisfy the CAR
requirements of users, has not been considered in the literature. In this section, we compare the fairness aspects
of those algorithms as they were all proposed for fair resource assignment in multi-user systems. Before we
investigate the fairness aspect in detail, we first introduce \textit{Liu's scheduling algorithm}~\cite{Liu01} and
\textit{the distribution fairness scheduling algorithm}~\cite{Qin04}.} In Liu's scheduling algorithm, BS selects a
user in each slot by using the following criterion:
\begin{equation}\label{eq:Liu}
\arg\max_{i \in \{1,2,..., n\}}[R_i(t)+c_i],
\end{equation}
where the offset $c_i$ is determined in order to satisfy the given CAR requirements. Liu's scheduling algorithm
maximizes the sum throughput for the given CAR requirements of users. In the distribution fairness scheduling
algorithm, the BS selects a user in each slot by using the following criterion:
\begin{equation}\label{eq:Qin}
\arg\max_{i \in \{1,2,..., n\}}[F_{i}(\gamma_i(t))+d_i],
\end{equation}
where the offset $d_i$ is determined in order to satisfy the given CAR requirements. The distribution fairness
scheduling algorithm maximizes the sum of the CDF values of the selected users.

All of CS, Liu's scheduling algorithm, and the distribution fairness scheduling algorithm satisfy the CAR
requirements, but they result in different throughput performance to users because of their diverse user selection
policies. Note that these three algorithms were originally proposed to address the fairness issue when exploiting
multi-user diversity in wireless communication systems. If fairness is defined as the satisfaction of the CAR
requirements, all three algorithms are equally fair. However, the different throughput performance of these
algorithms motivates us to reconsider the fairness issue for the scheduling algorithms that can satisfy the CAR
requirements. While CAR is apparently an important fairness criterion, it is not enough to capture all aspects of
fairness among users. On the other hand, the degree of achieved multi-user diversity can be another consideration
for fairness of users. A fair scheduling algorithm may aim at an identical degree of multi-user diversity for all
users. Users may feel unfair if the degrees of achieved multi-user diversity of users are different, even though
users satisfy their required CARs. In characterizing the degrees of achieved multi-user diversity, we take the
following two considerations:
\begin{itemize}
\item Exploiting multi-user diversity means that the BS selects a user when its channel has a high quality. Hence, a
criterion to measure the quality of assigned resource is required. The CDF value of SNR is a possible candidate
because it represents the quality of the channel gain with a real number in $[0, 1]$ and it is independent on the
average SNRs and the SNR distributions. We define $D(\alpha)$ as the average CDF value of a user given that the user is
selected, and it represents the quality of the assigned resource in this paper. {Then, it is expressed as
\begin{equation}
D(\alpha) =  \int_0^\infty F(\gamma) d F_{\rm sel}(\gamma),
\end{equation}
where $F_{\rm sel}(\gamma)$ is the SNR distribution given that the user is selected.} A larger value of $D(\alpha)$
indicates a better quality of assigned resource in the average sense. SNR itself (or data rate achieved from the
assigned resource) cannot be the index of quality of assigned resource because different users have different
average SNRs and different SNR distributions. Thus, directly comparing the SNR values of users results in
unfairness among users.

\item It is well known that a larger number of users in a cell provides a higher multi-user diversity. Since a
larger number of users can be interpreted as a smaller CAR for each user, a user with a smaller CAR has a higher
potential of exploiting multi-user diversity. Therefore, we take into account the different potentials from the
different CARs of users. {\em Lemma \ref{tm:UB}} gives us a guideline for characterizing this potential. It tells
us that the best quality of assigned resource for a given CAR $\alpha$ is obtained by selecting the user whose SNR is larger than $F^{-1}(1-\alpha)$. Let $D_{\rm UB}(\alpha)$ denote the upper-bound of the average CDF
value obtained by this optimal scheduling algorithm, which is given as
\begin{equation}
D_{\rm UB}(\alpha) = \frac{1}{\alpha}\int_{F^{-1}(1-\alpha)}^\infty F(\gamma) d F(\gamma) = \frac{1}{2}(2-\alpha),
\end{equation}
and obtained by replacing $R(\gamma)$ by $F(\gamma)$ in (\ref{eq:S_UB}).
\end{itemize}

The closeness of $D(\alpha)$ to $D_{\rm UB} (\alpha)$ means a higher degree of multi-user diversity is achieved.
Hence, we define the degree of achieved multi-user diversity for a user as the ratio of $D(\alpha)$ to $D_{\rm UB}
(\alpha)$. It is expressed as:
\begin{equation}\label{eq:QoR_2}
I_D (\alpha) = \frac{D(\alpha)}{D_{\rm UB}(\alpha)} \leq 1,
\end{equation}
{The upper-bound in (\ref{eq:QoR_2}) can be obtained by simply replacing $R(\gamma)$ by $F(\gamma)$ in
Appendix~\ref{apx:lm_UB}.} Given the CAR requirement $\alpha$ of a user, $I_D(\alpha)$ represents the degree of
achieved multi-user diversity with scheduling. A fair scheduling algorithm should provide similar values of $I_D$
for all users in spite of their diverse CAR requirements. If $I_D(\alpha)$ approaches 1, we can consider that the
scheduling algorithm optimally exploits multi-user diversity. Since the primary objective of the scheduling
algorithms considered in this paper is to exploit multi-user diversity, a good scheduling algorithm also maximizes
all users' $I_D$ values, i.e., maximizes
\begin{equation}\label{eq:quantitative_index}
I_{D,\min} = \min_{i \in \{1,2,..., n\}} {I_{D, i}}(\alpha_i).
\end{equation}
A good scheduling algorithm not only provides similar values of $I_D$ for all users, but also maximizes $I_{D,
\min}$. We define $I_{D,\min}$ as the QFI of a scheduling algorithm in this paper. It is notable that a QFI around
1 implicitly means that the $I_D$ values of all users are similar to each other because they have to be larger than
$I_{D,\min}$ and smaller than 1 by definition.

For systems with diverse users who require different CARs, now we can investigate the fairness among these users by
utilizing two aspects: {\em quantitative fairness} and {\em qualitative fairness}. Quantitative fairness indicates
the satisfaction of users CAR requirements, while qualitative fairness refers to the satisfaction on the quality of
the assigned resources to users. Qualitative fairness is closely related to the degrees of achieved multi-user
diversity for users. A fair scheduler should satisfy both criteria as much as possible. {Note that QFI in
(\ref{eq:quantitative_index}) is not the only fairness criterion to measure the degree of achieved multi-user
diversity, but is one possible candidate which is considered in this paper.} 

\begin{theorem} \label{tm:CDF_QF}
{If the CAR of a user is $\alpha$, the degree of achieved multi-user diversity with CS is given by }
\begin{equation}
I_D(\alpha) = \frac{2}{(1+\alpha)(2-\alpha)} \geq \frac{8}{9}.
\end{equation}
\end{theorem}
\begin{proof}
For a user with CAR of $\alpha$, the average CDF value given the user is selected with CS is calculated as
\begin{equation}
D(\alpha) =  \int_0^\infty F(\gamma) d [F(\gamma)]^{\frac{1}{\alpha}} =  \int_{0}^{1} u d u^{\frac{1}{\alpha}} =
\frac{1}{1 + \alpha}.
\end{equation}
Consequently, the corresponding $I_D$ value is calculated as
\begin{equation}\label{eq:I_D_min}
I_D(\alpha) = \frac{D(\alpha)}{D_{\rm UB}(\alpha)} = \frac{2}{(1+\alpha)(2-\alpha)} = \frac{2}{\frac{9}{4} -
(\alpha - \frac{1}{2})^2} \geq \frac{8}{9}.
\end{equation}
\end{proof}

It is easy to check that all the scheduling algorithms considered in this section strictly satisfy quantitative
fairness, but provide different qualitative fairness defined in (\ref{eq:QoR_2}) due to their different user
selection policies. We shall investigate the qualitative fairness of these scheduling algorithms in more detail in
Section~\ref{sec:Numerical_results}.

\section{Numerical Results}\label{sec:Numerical_results}

{We first investigate the throughput and fairness aspects of CS, Liu's scheduling algorithm and the distribution
fairness scheduling algorithm. From the user selection policies of the scheduling algorithms shown in
(\ref{eq:CDF_policy}), (\ref{eq:Liu}) and (\ref{eq:Qin}), we can easily check that they show identical throughput
performance when all users have the same CAR requirement, experience the same fading channel, and have the same
average SNR. To investigate their differences, we first consider a system with two asymmetric users: one user
experiences a Rayleigh fading channel and the other experiences an Nakagami-$m$ fading channel with $m=4$. The
average SNR of both users is set to 0dB and the sum of their CARs is 1.}

{Fig.~\ref{fig:Throughput} shows the sum throughput performance when the CAR of the first user is varied. We can
see that all the algorithms show better throughput performance than RRS. Liu's scheduling algorithm shows the best
sum throughput performance as it is designed for maximizing sum throughput. However, it does not mean that Liu's
scheduling algorithm maximizes individual users' throughput, which will be observed from
Fig.~\ref{fig:Throughput_gain}. Moreover, although there exists some differences, the CDF-based scheduling, Liu's
scheduling algorithm and the distribution fairness scheduling algorithm show similar sum throughputs.}
\begin{figure}
\centering
\includegraphics[width=3.5in]{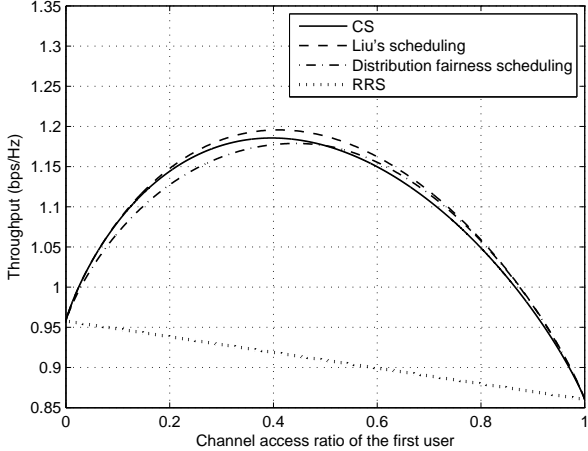}
    \caption{Sum throughput of the two users.}
    \label{fig:Throughput}
\end{figure}

{Fig.~\ref{fig:Throughput_gain} shows the throughput gain of each user when the user's CAR is varied. Note that the
x-axis label in Fig.~\ref{fig:Throughput_gain} is the CAR of the user being observed but not the first user's CAR. As all
the algorithms' throughput gains are larger than 1, they always show better performance than  RRS. Moreover, the
throughput gains of Liu's and the distribution fairness scheduling algorithms also increase as
the CAR decreases, which is the property of CS shown in {\em Theorem~\ref{tm:CDF_th}}. We can see that CS shows a
better throughput gain performance than Liu's and the distribution fairness scheduling
algorithms when the user's CAR becomes smaller. When a user's throughput gain of CS is larger than that of Liu's
scheduling algorithm, the other user's throughput gain of Liu's scheduling algorithm should be larger than that of
CS since Liu's scheduling algorithm maximizes sum throughput. Hence, Liu's scheduling algorithm does not always
provide the best throughput performance for all users. As the throughput performance of CS is independent from the
number of users contending for the channel and other users' channel statistics as indicated by (\ref{eq:S_CDF}),
the results of CS shown in Fig.~\ref{fig:Throughput_gain} are also valid for any stationary system where there
exists a user experiencing Nakagami-$m$ fading with $m=1,4$ and having the average SNR of 0dB.}

\begin{figure}
\centering
\includegraphics[width=3.5in]{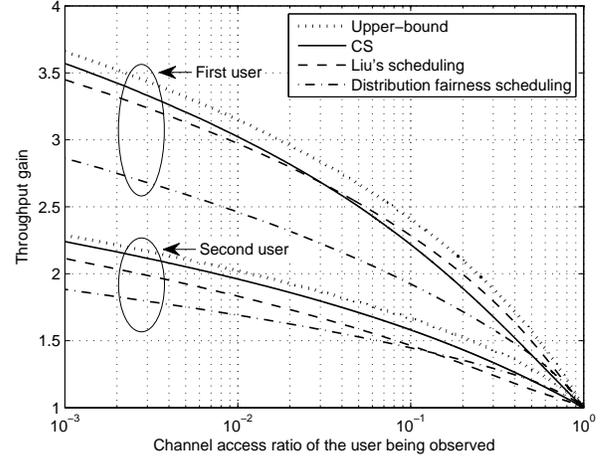}
    \caption{Throughput gain vs. channel access ratio}
    \label{fig:Throughput_gain}
\end{figure}

{From Fig.~\ref{fig:Throughput_gain}, we can observe that the throughput gain of CS is very close to the
upper-bound. Although we did not include the figure, we have also investigated the ratio between the throughput of
CS and the throughput upper-bound. CS can achieve at least 88\% and 93\% throughput performance compared to the
upper-bound when $m=1, 4$, respectively, for all values of the channel access ratios in [0, 1]. When the average
SNR is set to 10dB, it is observed that CS achieves at least 91\% and 95\% throughput performance compared to the
upper-bound throughput when $m=1, 4$, respectively.}

In order to investigate the fairness aspects, we plot Fig.~\ref{fig:I_D}, which shows the respective values of $I_D$ for RRS,
CS, Liu's scheduling algorithm and the distribution fairness scheduling algorithm when the CAR requirements of the
first and second users are 0.7 and 0.3. In this scenario, we can observe that Liu's and the distribution fairness scheduling algorithms favor the first user in exploiting multi-user diversity, whereas CS enables both users to exploit multi-user diversity in a more balanced manner. As discussed in Section~\ref{sec:Analysis}, Rayleigh fading provides a higher throughput gain compared to Nakagami-$m$ fading
channel with $m=4$. Hence, the first user is able to better exploit multi-user diversity and contribute more additional throughput than the second user with Liu's scheduling algorithm, which
aims to maximize the total throughput under the CAR constraints. This is why Liu's scheduling algorithm shows
the largest gap between the two users among the scheduling
algorithms considered. As expected, RRS shows the worst performance. {As RRS does not exploit multi-user diversity,
$D(\alpha)$ is always constant at 1/2. Since $D_{\rm UB} (\alpha)$ decreases as $\alpha$ increases, the first user
who requires a higher CAR shows a better $I_D(\alpha)$ value than the second user.}

\begin{figure}
\centering
\includegraphics[width=3.5in]{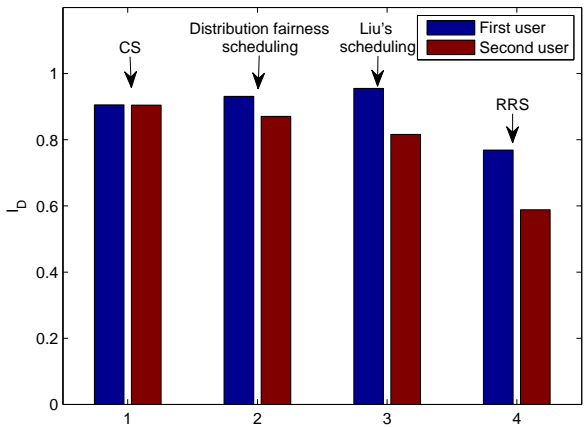}
    \caption{$I_{D}$ for the two users.}
    \label{fig:I_D}
\end{figure}

Fig. \ref{fig:I_CDF} shows the values of $I_{D, \min}$ for CS, Liu's scheduling algorithm, the distribution
fairness scheduling algorithm and RRS for varying CAR of the first user. CS yields relatively good qualitative
fairness for any CAR requirement as shown in Fig.~\ref{fig:I_CDF}. Notably, CS shows a predictable lower-bound of
$I_{D, \min}$, $8/9$, as proven by {\em Theorem~\ref{tm:CDF_QF}} while neither Liu's scheduling algorithm nor the
distribution fairness scheduling algorithm can guarantee any lower-bound of $I_{D, \min}$, which varies over
the number of competing users and the users' CAR requirements. {From the viewpoint of qualitative fairness, RRS
shows the worst performance as it does not exploit multi-user diversity and the value of $I_{D, \min}$ ranges
between $[\frac{1}{2}, \frac{2}{3}]$.}

\begin{figure}
\centering
\includegraphics[width=3.5in]{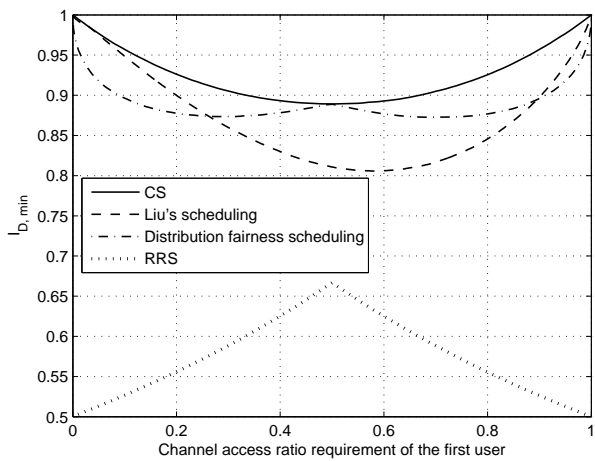}
    \caption{$I_{D, \min}$ vs. the channel access ratio requirement of the first user.}
    \label{fig:I_CDF}
\end{figure}

{Fig.~\ref{fig:Throughput_gain_SNR} shows the throughput gains of the two users with CS, Liu's scheduling and the
distribution fairness scheduling algorithm when both users' average SNRs varied from $0$dB to $20$dB and the CAR is
set to $0.1$ for the user being observed. The throughput gain decreases as the average SNR or the shape parameter
increases as shown in {\em Remark \ref{lm:scale_nakagami_2}}. We can observe the tradeoff between CS and Liu's
scheduling algorithm over different SNRs and CARs.}
\begin{figure}
\centering
\includegraphics[width=3.5in]{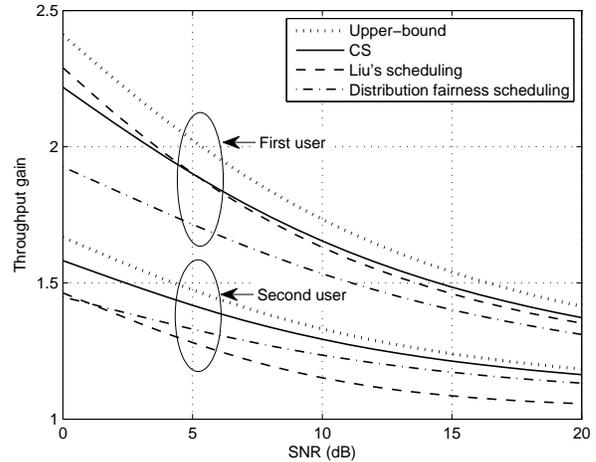}
    \caption{Throughput gain vs. average SNR.}
    \label{fig:Throughput_gain_SNR}
\end{figure}

Fig.~\ref{fig:Number_feedback} shows the average feedback overhead for when the NFB probability is varied from $0$
to $1$ and all the users are equally weighted. The average feedback overhead is defined as the number of users who
send feedback information to the BS. The feedback overhead of CS-FR decreases as the NFB probability increases.
From the figure, we can observe that a large feedback overhead is required as the number of users in a cell
increases for a given NFB probability. If the NFB probability is 2\%, i.e., $p = 0.02$, the average feedback
overhead are $2.71$, $3.24$, $3.84$, and $3.91$ when there are 5, 10, 100, and $\infty$ users in a cell,
respectively.

\begin{figure}
\centering
\includegraphics[width=3.5in]{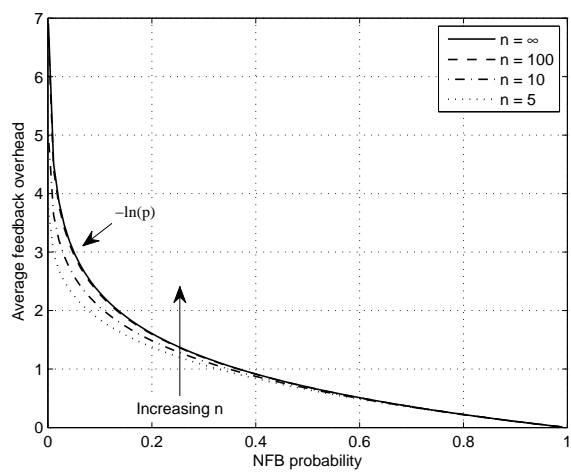}
    \caption{Average feedback overhead vs. NFB probability}
    \label{fig:Number_feedback}
\end{figure}

Fig.~\ref{fig:Feedback_ratio} shows the feedback overhead ratios with equally weighted users when the NFB
probability is varied from 0 to 1. The average feedback ratio represents the ratio of the average number of users
sending feedback information to BS with CS-FR over the total number of users. Note that this equally weighted case
yields an upper-bound for the unequally weighted case as discussed in Section \ref{subsec:feedback_overhead}. From
the figure we can observe that a larger NFB probability reduces the feedback overhead more significantly. If the
NFB probability is 2\%, i.e., $p = 0.02$, the average feedback ratio is equal to 54.3\%, 32.4\%, and 3.8\% when
$n=5,10,100$, respectively. Therefore, for given a NFB probability, CS-FR reduces the feedback overhead
significantly as the number of users increases.

\begin{figure}
\centering
    \includegraphics[width=3.5in]{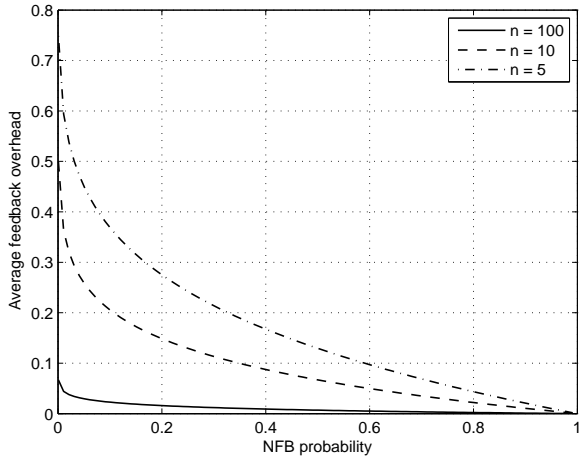}
    \caption{Average feedback ratio vs. NFB probability}
    \label{fig:Feedback_ratio}
\end{figure}

Fig.~\ref{fig:throughput_gain_m} shows the throughput gain of CS and CS-FR for varying $1/\alpha$. For equally
weighted users, $1/\alpha$ is equal to the number of users in the system. The average SNR of the user
being observed is set to 0dB. We can observe that the throughput gain of CS-FR increases as $1/\alpha$
increases and a larger NFB probability reduces the throughput gain with CS-FR. In Nakagami-$m$ fading channels,
CS-FR yields a larger throughput gain with small $m$ compared with CS since users experiencing more channel
fluctuations obtain a higher throughput gain compared to a user with less fluctuations.

\begin{figure}
\centering
    \includegraphics[width=3.5in]{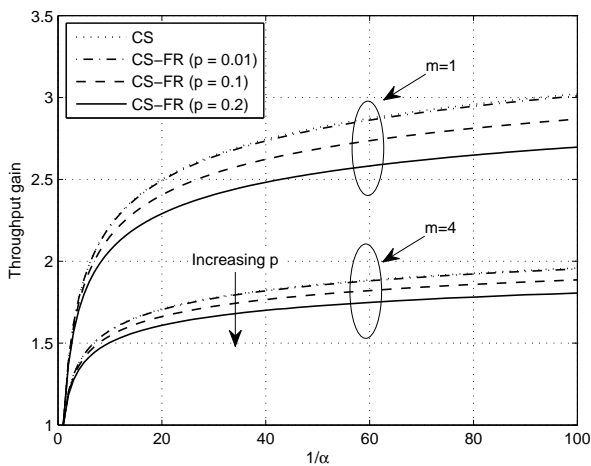}
    \caption{Throughput gain vs. $1/\alpha$}
    \label{fig:throughput_gain_m}
\end{figure}

Fig.~\ref{fig:S_WFS_FR} shows the throughput gains of CS-FR for various NFB probabilities over a Rayleigh fading channel, which is a special case of a Nakagami-$m$ fading channel with $m=1$, with the average SNR = 0dB. We
can observe that a smaller CAR and a smaller NFB probability yield a larger throughput gain.
Fig.~\ref{fig:S_WFS_FR_ratio} shows the throughput ratio between CS-FR and CS in
the same environment. We can observe that a smaller CAR yields a smaller value of throughput ratio. Thus, if CS-FR
is applied, a user with a smaller CAR is more prone to a throughput loss compared to a user with a larger
CAR. A similar trend can also be observed from the lower-bound throughput of CS-FR shown in (\ref{eq:FR_Low}) since
the formula, $1 - p+\alpha p^{1-\alpha}$, is an increasing function of $\alpha$.

\begin{figure}
    \centering
    \includegraphics[width=3.5in]{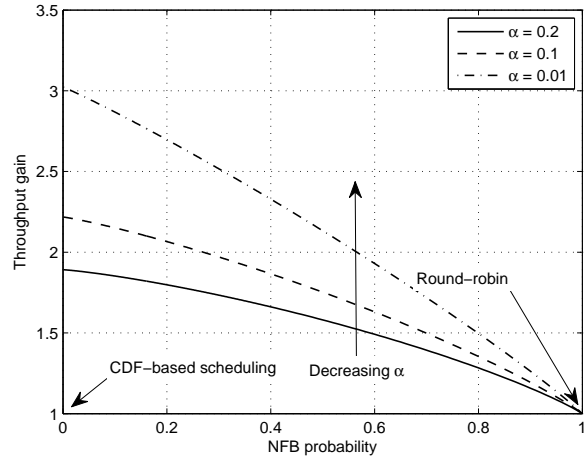}
    \caption{Throughput gain vs. NFB probability.}
    \label{fig:S_WFS_FR}
\end{figure}

\begin{figure}
    \centering
    \includegraphics[width=3.5in]{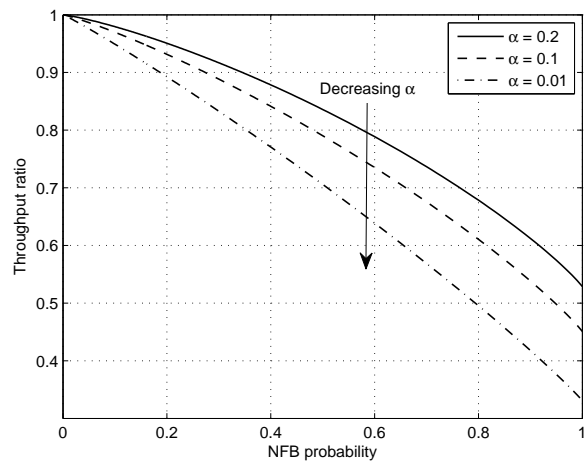}
    \caption{Ratio between the throughput values with CS-FR and CS.}
    \label{fig:S_WFS_FR_ratio}
\end{figure}

\section{Conclusions}\label{sec:Conclusions}
In this paper, we have investigated the fundamental performance limits of CS in terms of throughput, fairness, and
feedback overhead. We have rigorously characterized the throughput behavior of CS. The throughput upper-bound of
general schedulers for a given SNR distribution, data rate function, and channel access ratio has been derived and
CS has been proven to achieve the upper-bound when the data rate function has the upper limit and the CAR decreases
to zero. The lower- and upper-bound of the throughput gain with CS have also been analyzed. We have further
proposed CS-FR, a novel feedback reduction technique for CS. With CS-FR, a single threshold is
sufficient to satisfy the diverse CARs of all users, and the feedback overhead is upper-bounded by $-\ln p$ where
$p$ represents the probability that no user satisfies the threshold condition. We have also investigated the
throughput characteristics and observed that the throughput loss due to feedback reduction relative to the
throughput with full feedback is upper-bounded by $-\ln p.$ Finally, we have proposed the concept of qualitative
fairness in order to more thoroughly investigate fairness among various schedulers, and shown that CS achieves
relatively better qualitative fairness, compared with the other existing scheduling algorithms.

\appendices
\section{Properties of $S(x, \alpha)$}{\label{apx:Property}}
\subsection{Proof for Property \ref{pr:alpha_inc}:}
Let $0< \alpha_1 < \alpha_2\leq 1$, then we have
\begin{equation}
\begin{array}{ll}
\displaystyle \alpha_1 S(x, \alpha_1)
            &=    \int_{F^{-1}(x)}^\infty R(\gamma)[F(\gamma)]^{\frac{1}{\alpha_1}-1}d[F(\gamma)]\\
            &=    \alpha_2 \int_{F^{-1}(x)}^\infty R(\gamma)\frac{1}{\alpha_2}[F(\gamma)]^{\frac{1}{\alpha_1}-1}d[F(\gamma)]\\
            &\leq \alpha_2\int_{F^{-1}(x)}^\infty R(\gamma)
            \frac{1}{\alpha_2}[F(\gamma)]^{\frac{1}{\alpha_2}-1}d[F(\gamma)]\\
            & =    \alpha_2 S(x, \alpha_2).
\end{array}
\end{equation}
The inequality comes from the fact that $ F( \gamma)^{\frac{1}{\alpha_1}} \leq F(
\gamma)^{\frac{1}{\alpha_2}}$ for $0<F(\gamma)<1$.

\subsection{Proof for Property \ref{pr:alpha_dec}:}
Let $0< \alpha_1 < \alpha_2\leq 1$. Then, we have
\begin{equation}
\begin{array}{ll}
 S(x^{\alpha_1}, \alpha_1) =\int_{F^{-1}(x^{\alpha_1})}^{\infty} R(\gamma)d[F(\gamma)]^{\frac{1}{\alpha_1}} \\
 =\int_{x}^{1} R(F^{-1}(u^{\alpha_1}))du ~~~~~~~~~ ({\rm Replacing~}u = [F(\gamma)]^{\frac{1}{\alpha_1}}.)\\
 \geq \int_{x}^{1} R(F^{-1}(u^{\alpha_2}))du \\
 =\int_{x^{\alpha_2}}^{1} R(F^{-1}(v))dv^{\frac{1}{\alpha_2}} ~~~~~~~ ({\rm Replacing~} v =u^{\alpha_2}.)\\
 =\int_{F^{-1}(x^{\alpha_2})}^{\infty} R(\gamma)d[F(\gamma)]^{\frac{1}{\alpha_2}}~~({\rm Replacing~} F(\gamma) = v.)\\
= S(x^{\alpha_2}, \alpha_2),
\end{array}
\end{equation}
where we have applied the increasing property of $R(F^{-1}(y))$ with $u^{\alpha_1}
>u^{\alpha_2}.$

\subsection{Proof for Property \ref{pr:p_dec}:}
For $0 \leq x_1 \leq x_2 \leq 1$, we have
\begin{equation}
\begin{array}{ll}
 S(x_1^{\alpha}, \alpha)&=\int_{F^{-1}(x_1^{\alpha})}^{\infty}
 R(\gamma)d[F(\gamma)]^{\frac{1}{\alpha}}\\
&\geq \int_{F^{-1}(x_2^{\alpha})}^{\infty} R(\gamma)d[F(\gamma)]^{\frac{1}{\alpha}} =
S(x_2^{\alpha}, \alpha),
\end{array}
\end{equation}
where we have used the fact that $F^{-1}(x_1^{\alpha}) \leq F^{-1}(x_2^{\alpha})$.

\subsection{Proof for Property \ref{pr:p_inc}:} For $0 \leq x_1 \leq x_2 \leq 1$, we have
\begin{equation}
\begin{array}{ll}
S(x_1, \alpha)\\
 =    \int_{F^{-1}(x_2)}^\infty R(\gamma)d[F(\gamma)]^\frac{1}{\alpha} + \int_{F^{-1}(x_1)}^{F^{-1}(x_2)} R(\gamma)d[F(\gamma)]^\frac{1}{\alpha} \\
 \leq \int_{F^{-1}(x_2)}^\infty R(\gamma)d[F(\gamma)]^\frac{1}{\alpha} \!+\! \int_{F^{-1}(x_1)}^{F^{-1}(x_2)} R(F^{-1}(x_2))d[F(\gamma)]^\frac{1}{\alpha} \\
 =    \int_{F^{-1}(x_2)}^\infty R(\gamma)d[F(\gamma)]^\frac{1}{\alpha} + (x_2^\frac{1}{\alpha} - x_1^\frac{1}{\alpha}) R(F^{-1}(x_2)) \\
 =    \int_{F^{-1}(x_2)}^\infty R(\gamma)d[F(\gamma)]^\frac{1}{\alpha} \\
 ~~~~ ~~~~+ \frac{x_2^\frac{1}{\alpha} - x_1^\frac{1}{\alpha}}{1 - x_2^\frac{1}{\alpha}} \int_{F^{-1}(x_2)}^\infty R(F^{-1}(x_2))d[F(\gamma)]^\frac{1}{\alpha}  \\
 \leq \int_{F^{-1}(x_2)}^\infty R(\gamma)d[F(\gamma)]^\frac{1}{\alpha}  \\
 ~~~~ ~~~~+ \frac{x_2^\frac{1}{\alpha} - x_1^\frac{1}{\alpha}}{1 - x_2^\frac{1}{\alpha}} \int_{F^{-1}(x_2)}^\infty R(\gamma)d[F(\gamma)]^\frac{1}{\alpha}  \\
 =    \frac{1 - x_1^\frac{1}{\alpha}}{1 - x_2^\frac{1}{\alpha}} \int_{F^{-1}(x_2)}^\infty
R(\gamma)d[F(\gamma)]^\frac{1}{\alpha} \\
= \frac{1 - x_1^\frac{1}{\alpha}}{1 - x_2^\frac{1}{\alpha}} S(x_2, \alpha).
\end{array}
\end{equation}
The second inequality comes from the increasing property of $R(\gamma)$.

\subsection{Proof for Property \ref{pr:SL_alpha_inc}}
Let $0< \alpha_1 < \alpha_2\leq 1$, then we have
\begin{equation}
\begin{array}{ll}
 \alpha_1 S_L(x, \alpha_1)
            &=    \alpha_2 \int_0^{F^{-1}(x)} R(\gamma)\frac{1}{\alpha_2}[F(\gamma)]^{\frac{1}{\alpha_1}-1}d[F(\gamma)]\\
            &\leq \alpha_2\int_0^{F^{-1}(x)} R(\gamma) \frac{1}{\alpha_2}[F(\gamma)]^{\frac{1}{\alpha_2}-1}d[F(\gamma)] \\
            &=    \alpha_2 S_L(x, \alpha_2),
\end{array}
\end{equation}

\subsection{Proof for Property \ref{pr:FR_p_dec}}
For the decreasing property, we show that the derivative of the function is smaller than or
equal to 0.
\begin{equation}
\begin{array}{l}
\frac{d}{dx}\{S(x^{\alpha}, \alpha) + x^{1-\alpha}S_L(x^{\alpha},
1)\} \\
= \frac{d}{dx} \{\int_{ x^{\alpha}}^1 R(F^{-1}(u)) d u^{\frac{1}{\alpha}} +
{x^{1-\alpha}} \int_0^{ x^{\alpha}} R(F^{-1}(u)) d u \} \\
=       - (1-\alpha)R(F^{-1}(x^{\alpha})) \!   +\! (1\!-\!\alpha){ x^{-\alpha}} \int_0^{ x^{\alpha}} R(F^{-1}(u)) d u\\
\leq    - (1-\alpha)R(F^{-1}(x^{\alpha})) \!   +\! (1\!-\!\alpha){ x^{-\alpha}} \int_0^{
x^{\alpha}}
R(F^{-1}( x^{\alpha})) d u\\
=0,
\end{array}
\end{equation}
where the inequality comes from the increasing property of the functions $R(\gamma)$ and
$F(\gamma)$.

\section{Proof of Lemma \ref{tm:UB}}\label{apx:lm_UB}

Let $g(\gamma)$~($0\leq g(\gamma)\leq 1$) be the selection probability where $\gamma$
indicates the SNR of the user. The CAR is equal to $\alpha$ and we have
\begin{equation}
\int_0^\infty g(\gamma) dF (\gamma) = \alpha.
\end{equation}
Then, the achievable throughput of the user with the CAR is expressed as:
\begin{equation}
\begin{array}{ll}
\int_0^\infty R(\gamma)g(\gamma) dF(\gamma) \\
=    \int_{F^{-1}(1-\alpha)}^\infty R(\gamma)g(\gamma)dF(\gamma) + \int_0^{F^{-1}(1-\alpha)} R(\gamma)g(\gamma)dF(\gamma)\\
\leq \int_{F^{-1}(1-\alpha)}^\infty R(\gamma)g(\gamma)dF(\gamma) \\+ R({F^{-1}(1-\alpha)}) \int_0^{F^{-1}(1-\alpha)} g(\gamma)dF(\gamma)\\
=    \int_{F^{-1}(1-\alpha)}^\infty R(\gamma)g(\gamma)dF(\gamma) \\+ R({F^{-1}(1-\alpha)}) [\int_0^\infty g(\gamma)dF(\gamma) - \int_{F^{-1}(1-\alpha)}^{\infty} g(\gamma)dF(\gamma)]\\
=    \int_{F^{-1}(1-\alpha)}^\infty R(\gamma)g(\gamma)dF(\gamma) \\+ R({F^{-1}(1-\alpha)}) [\alpha- \int_{F^{-1}(1-\alpha)}^{\infty} g(\gamma)dF(\gamma)]\\
=    \int_{F^{-1}(1-\alpha)}^\infty R(\gamma)g(\gamma)dF(\gamma) \\+ R({F^{-1}(1-\alpha)}) [\int_{F^{-1}(1-\alpha)}^{\infty} dF(\gamma)\!-\! \int_{F^{-1}(1-\alpha)}^{\infty} \!g(\gamma)dF(\gamma)]\\
=    \int_{F^{-1}(1-\alpha)}^\infty R(\gamma)g(\gamma)dF(\gamma) \\+ R({F^{-1}(1-\alpha)}) \int_{F^{-1}(1-\alpha)}^{\infty} [1-g(\gamma)]dF(\gamma)\\
\leq \int_{F^{-1}(1-\alpha)}^\infty R(\gamma)g(\gamma)dF(\gamma) \\+  \int_{F^{-1}(1-\alpha)}^{\infty} R(\gamma)[1-g(\gamma)]dF(\gamma) \\
=    \int_{F^{-1}(1-\alpha)}^{\infty} R(\gamma)dF(\gamma) \\
= S(1 - \alpha, 1),
\end{array}
\end{equation}
where the two inequalities are come from the increasing property of $R(\gamma)$. If we replace $g(\gamma)$ with
$F(\gamma)^{\frac{1}{\alpha} - 1}$, we can observe that this upper-bound is always larger than the throughput of
CS.

\section{Proof of Theorem \ref{tm:UB_achivability}}\label{apx:lm_UB_achivability}

First, from \textit{Lemma \ref{tm:UB}}, we have
\begin{equation}\label{eq:lm_UB_1}
\displaystyle \lim_{\alpha \rightarrow 0} \frac{\alpha S(0, \alpha)}{S(1 -\alpha, 1)}\leq
1.
\end{equation}
On the other hand, we have {
\begin{equation}\label{eq:lm_UB_2}
\begin{array}{ll}
\lim\limits_{\alpha \rightarrow 0} \frac{\alpha S(0, \alpha)}{S(1 -\alpha, 1)} & = \lim\limits_{\alpha \rightarrow
0}  \frac{ \alpha \int_{0}^\infty R(\gamma)d[F(\gamma)]^{\frac{1}{\alpha}}}{\int_{F^{-1}(1-\alpha)}^\infty R(\gamma)dF(\gamma)} \\
& \geq  \lim\limits_{\alpha \rightarrow 0} \frac{\alpha\int_{\gamma_{0}}^\infty R(\gamma)d[F(\gamma)]^{\frac{1}{\alpha}}}{\int_{F^{-1}(1-\alpha)}^\infty R(\gamma)dF(\gamma)} \\
&=  \lim\limits_{\alpha \rightarrow 0}\frac{\alpha R_{\rm th}\{1 -
[F(\gamma_{0})]^{\frac{1}{\alpha}}\}}{ \alpha R_{\rm th}}\\
&= {\lim\limits_{\alpha \rightarrow 0}1 - [F(\gamma_{0})]^{\frac{1}{\alpha}}} \\
&= 1.
\end{array}
\end{equation}}
Comparing (\ref{eq:lm_UB_1}) and (\ref{eq:lm_UB_2}), we can conclude the statement.

\section{Proof of Theorem \ref{tm:low_up}}\label{apx:lm_low_up}
The upper bound is given by the \textit{Lemma \ref{tm:UB}} . For the lower bound, we have
\begin{equation}
\begin{array}{ll}
S_{\rm CS}(\alpha) &= \alpha \int_{0}^\infty R(\gamma)d[F(\gamma)]^{\frac{1}{\alpha}}   \\
                    &\geq \alpha \int_{F^{-1}(1 +\alpha\ln \alpha)}^\infty R(\gamma)d[F(\gamma)]^{\frac{1}{\alpha}} \\
                    &\geq \alpha R(F^{-1}(1  + \alpha\ln\alpha)) \int_{F^{-1}(1  + \alpha\ln \alpha)}^\infty \!\!d[F(\gamma)]^{\frac{1}{\alpha}} \\
                    &=    \alpha R(F^{-1}(1  + \alpha\ln\alpha)) [ 1 - (1  + \alpha\ln \alpha)^{\frac{1}{\alpha}} ].
\end{array}
\end{equation}
The throughput gains of the upper and lower bound throughput are calculated as
\begin{equation}\label{eq:UB_scale1}
\begin{array}{ll}
g_{\rm UB} &=\lim\limits_{\alpha \rightarrow 0}\frac{S_{\rm UB}(\alpha)}{S_{\rm RRS}(\alpha)}\\
           & = \lim\limits_{\alpha \rightarrow 0}\frac{\int_{1-\alpha}^1 R(F^{-1}(u))du}{\alpha\int_{0}^\infty
           R(\gamma)dF(\gamma)}~~~~~~~~~~~~~~~~~~\\
           & = \lim\limits_{\alpha \rightarrow 0}\frac{ R(F^{-1}(1-\alpha))}{E[R]},
\end{array}
\end{equation}
\begin{equation}\label{eq:UB_scale2}
\begin{array}{ll}
g_{\rm LB} &= \lim\limits_{\alpha \rightarrow 0}\frac{S_{\rm LB}(\alpha)}{S_{\rm RRS}(\alpha)}\\
           & = \lim\limits_{\alpha \rightarrow 0}\frac{\alpha R(F^{-1}(1  + \alpha\ln\alpha)) [ 1 - (1 + \alpha\ln \alpha)^{\frac{1}{\alpha}} ]}{\alpha\int_{0}^\infty
           R(\gamma)dF(\gamma)}\\
           & = \lim\limits_{\alpha \rightarrow 0}\frac{ R(F^{-1}(1+\alpha\ln\alpha))}{E[R]}.
\end{array}
\end{equation}
Here, we have used the property that $\lim\limits_{\alpha \rightarrow 0} (1 +\alpha \ln
\alpha)^{\frac{1}{\alpha}} = 0$.

\section{SNR distribution of the selected user with CS-FR }\label{apx:F_i_sel}
Given that the $i$-th user is selected, its SNR distribution in the NFB slots is derived as
\begin{equation}
\begin{array}{ll}
&F_{i, \rm Sel, NFB}(\gamma) \\& = \Pr\{ \gamma_i < \gamma|\text{user $i$ is selected, the slot is a  NFB slot}\}\\
                        & =  \frac{\Pr\{ \gamma_i < \gamma, \text{user $i$ is selected, the slot is a NFB slot}\}}{\Pr\{\text{user $i$ is selected, the slot is a NFB slot}\}}\\
                        & = \frac{\alpha_i \Pr\left\{ \gamma_i < \gamma, ~U_j^{\frac{1}{w_j}}<\eta_{\rm th},~\forall j \in \{1, 2, \cdots, n\}\right\}}{\alpha_ip}\\
                        & = \frac{\alpha_i \frac{p}{p^{\alpha_i}}\Pr\left\{ \gamma_i < \gamma, ~\gamma_i<F_{i}^{-1}(p^{\alpha_i})\right\}}{\alpha_ip}\\
                        & = \left\{
                        \begin{array}{ll}
                        p^{-\alpha_i} F_i(\gamma),&~~~~{\rm if}~0<\gamma<F_{i}^{-1}(p^{\alpha_i}),\\
                        1,                                 &~~~~ {\rm if~}\gamma \geq F_i^{-1}(p^{\alpha_i}),
                        \end{array}

                        \right.
\end{array}
\end{equation}
where we have used the fact $\eta_{\rm th}^{w_i}={p^{\alpha_i}} $ from (\ref{eq:P_out}). The
SNR distribution in the FB slots is derived as
\begin{equation}
\begin{array}{ll}
&F_{i, \rm Sel, FB}(\gamma) \\& = \Pr\{ \gamma_i < \gamma|\text{user $i$ is selected, the slot is a FB slot}\}\\
                        & = \frac{\Pr\{ \gamma_i < \gamma, \text{user $i$ is selected, the slot is a FB slot}\}}{\Pr\{\text{user $i$ is selected, the slot is a FB slot}\}}\\
                        & = \frac{\Pr\{ \gamma_i < \gamma, ~U_i^{\frac{1}{w_i}} > \eta_{\rm th},~ U_i^{\frac{1}{w_i}} > U_j^{\frac{1}{w_j}}, \forall j \in \{1, 2, \cdots, i-1, i+1, \cdots n\} \}}{\alpha_i(1-p)}\\
                        & = \frac{\Pr\{ \gamma_i < \gamma, ~\gamma_i > F_i^{-1}({p^{\alpha_i}}),~ U_j < [F_i(\gamma_i)]^{\frac{w_j}{w_i}} , \forall j \in \{1, 2, \cdots, n\} \& j \neq i\}}{\alpha_i(1-p)}\\
                        & = \frac{\int_{F_i^{-1}({p^{\alpha_i}})}^\gamma [F_i(\gamma_i)]^{\sum_{j=1,j\neq i}^n\frac{w_j}{w_i}} d F_i(\gamma_i)}{\alpha_i(1-p)}\\
                        & = \left\{
                        \begin{array}{ll}
                        0,&~~~~ {\rm if~}0<\gamma < F_i^{-1}({p^{\alpha_i}}),\\
                        \frac{[F_i(\gamma)]^{\frac{1}{\alpha_i}} - p}{1-p}, &~~~~ {\rm if~}\gamma \geq F_i^{-1}(p^{\alpha_i}).
                        \end{array}
                        \right.
\end{array}
\end{equation}
Finally, the SNR distribution given that the $i$-th user is selected is derived as
\begin{equation}
\begin{array}{l}
F_{i, \rm Sel}(\gamma)   = F_{i, \rm Sel, NFB}(\gamma) \Pr\{\text{NFB slot}\} \\
        ~~~~~~~~~~~~~~~~~~+ F_{i, \rm Sel, FB}(\gamma) \Pr\{\text{FB slot}\}\\
                    ~~ = F_{i, \rm Sel, NFB}(\gamma) p+ F_{i, \rm Sel, FB}(\gamma) (1-p)\\
                    ~~= \left\{
                        \begin{array}{ll}
                        p^{1-\alpha_i} F_i(\gamma),&~~~~ {\rm if~}0<\gamma < F_i^{-1}({p^{\alpha_i}}),\\
                        F_i(\gamma)^{\frac{1}{\alpha_i}}, &~~~~ {\rm if~}\gamma \geq F_i^{-1}(p^{\alpha_i}).
                        \end{array}
                        \right.
\end{array}
\end{equation}

\section{Proof for Theorem~\ref{lm:FR_Low}} \label{apx:lm_FR_Low}
\begin{equation}
\begin{array}{l}
\frac{1}{\alpha} S_{\rm CS-FR}(\alpha, p)
            = S(p^{\alpha}, \alpha) + p^{1-\alpha}S_L(p^{\alpha}, 1)\\
~~~~~~        \geq  S(p^{\alpha}, \alpha) + p^{1-\alpha}\alpha S_L(p^{\alpha}, \alpha)\\
~~~~~~        =  (1-\alpha p^{1-\alpha}) S(p^{\alpha}, \alpha) + \alpha p^{1-\alpha }  S(0, \alpha) \\
~~~~~~        \geq (1-\alpha p^{1-\alpha }) (1- p)S(0, \alpha) + \alpha p^{1-\alpha }  S(0, \alpha)\\
~~~~~~        = (1-p + \alpha p^{2-\alpha }) S(0, \alpha)\\
~~~~~~       \geq (1 -  p )S(0, \alpha)\\
~~~~~~       = (1 -p )\frac{1}{\alpha}S_{\rm CS}(\alpha).
\end{array}
\end{equation}
where {\em Property \ref{pr:SL_alpha_inc}} and {\em Property~\ref{pr:p_inc}} have been
applied to obtain the first and second inequalities, respectively.

\bibliographystyle{./References/IEEEtran}
\bibliography{./References/IEEEabrv,./References/RefFull,./References/References}

\begin{thebibliography}{10}
\providecommand{\url}[1]{#1}
\csname url@samestyle\endcsname
\providecommand{\newblock}{\relax}
\providecommand{\bibinfo}[2]{#2}
\providecommand{\BIBentrySTDinterwordspacing}{\spaceskip=0pt\relax}
\providecommand{\BIBentryALTinterwordstretchfactor}{4}
\providecommand{\BIBentryALTinterwordspacing}{\spaceskip=\fontdimen2\font plus
\BIBentryALTinterwordstretchfactor\fontdimen3\font minus
  \fontdimen4\font\relax}
\providecommand{\BIBforeignlanguage}[2]{{%
\expandafter\ifx\csname l@#1\endcsname\relax
\typeout{** WARNING: IEEEtran.bst: No hyphenation pattern has been}%
\typeout{** loaded for the language `#1'. Using the pattern for}%
\typeout{** the default language instead.}%
\else
\language=\csname l@#1\endcsname
\fi
#2}}
\providecommand{\BIBdecl}{\relax}
\BIBdecl

\bibitem{Knopp95}
R.~Knopp and P.~A. Humblet, ``{Information capacity and power control in single
  cell multiuser communications},'' \emph{in Proc. {IEEE} International
  Conference on Communications (ICC)}, pp. 331--335, Jun. 1995.

\bibitem{Tse97}
D.~N.~C. Tse, ``{Optimal power allocation over parallel Gaussian broadcast
  channels},'' \emph{in Proc. {IEEE} International Symposium on Information
  Theory (ISIT)}, p.~27, Jun.-Jul. 1997.

\bibitem{Borst01}
S.~Borst and P.~Whiting, ``{Dynamic rate control algorithms for HDR throughput
  optimization},'' \emph{in Proc. {IEEE} International Conference on Computer
  Communications (INFOCOM)}, pp. 976--985, Aug. 2001.

\bibitem{Chung07}
S.-Y. Chung and P.~A. Humblet, ``{On the robustness of scheduoling against
  channel variations},'' \emph{{IEEE} Trans. Wireless Commun.}, vol.~6, no.~9,
  pp. 3186--3190, Sep. 2007.

\bibitem{Viswanath02}
P.~Viswanath, D.~N.~C. Tse, and R.~Laroia, ``{Opportunistic beamforming using
  dumb antennas},'' \emph{{IEEE} Trans. Inf. Theory}, vol.~48, no.~6, pp.
  1277--1294, Jun. 2002.

\bibitem{Kim05}
H.~Kim and Y.~Han, ``{A proportional fair scheduling for multicarrier
  transmission systems},'' \emph{{IEEE} Commun. Lett.}, vol.~9, no.~3, pp.
  210--212, Mar. 2005.

\bibitem{Liu10}
E.~Liu and K.~K. Leung, ``{Expected throughput of the proportional fair
  scheduling over Rayleigh fading channels},'' \emph{{IEEE} Commun. Lett.},
  vol.~14, no.~6, pp. 515--517, Jun. 2010.

\bibitem{Liu01}
X.~Liu, E.~K.~P. Chong, and N.~B. Shroff, ``{Opportunistic transmission
  scheduling with resource-sharing constraints in wireless networks},''
  \emph{{IEEE} J. Sel. Areas Commun.}, vol.~19, no.~10, pp. 2053--2064, Oct.
  2001.

\bibitem{Sharif05}
M.~Sharif and B.~Hassibi, ``{On the capacity of MIMO broadcast channels with
  partial side information},'' \emph{{IEEE} Trans. Inf. Theory}, vol.~51,
  no.~2, pp. 506--522, Feb. 2005.

\bibitem{Chen06}
C.-J. Chen and L.-C. Wang, ``{A unified capacity analysis for wireless systems
  with joint multiuser scheduling and antenna diversity and Nakagami fading
  channels},'' \emph{{IEEE} Trans. Commun.}, vol.~54, no.~3, pp. 469--478, Mar.
  2006.

\bibitem{Hwang08}
G.~U. Hwang and F.~Ishizaki, ``{Design of a fair scheduler exploiting multiuser
  diversity with feedback information reduction},'' \emph{{IEEE} Commun.
  Lett.}, vol.~12, no.~2, pp. 124--126, Feb. 2008.

\bibitem{Kelly98}
F.~P. Kelly, A.~K. maulloo, and D.~K.~H. Tan, ``{Rate control in communication
  networks: shadow prices, proportional fairness and stability},''
  \emph{Journal of the Operational Research Society}, vol.~49, pp. 237--252,
  1998.

\bibitem{Shreedhar96}
M.~Shreedhar and G.~Varghese, ``{Efficient fair queuing using deficit
  Round-robin},'' \emph{{IEEE/ACM} Trans. Netw.}, vol.~4, no.~3, pp. 375--385,
  Jun. 1996.

\bibitem{Sharma05}
N.~Sharma and L.~H. Ozarow, ``{A study of opportunism for multiple-antenna
  systems},'' \emph{{IEEE} Trans. Inf. Theory}, vol.~51, no.~5, pp. 1804--1814,
  May 2005.

\bibitem{Yang06}
L.~Yang and M.-S. Alouini, ``{Performance analysis of multiuser selection
  diversity},'' \emph{{IEEE} Trans. Veh. Technol.}, vol.~55, no.~6, pp.
  1848--1861, Nov. 2006.

\bibitem{Park05}
D.~Park, H.~Seo, H.~Kwon, and B.~G. Lee, ``{Wireless packet scheduling based on
  the cumulative distribution function of user transmission rate},''
  \emph{{IEEE} Trans. Commun.}, vol.~53, no.~11, pp. 1919--1929, 2005.

\bibitem{Qin04}
X.~Qin and R.~Berry, ``{Opportunistic splitting algorithms for wireless
  networks with hegerogeneous users},'' \emph{Proc. Conf. Inform. Sciences
  Systems (CISS)}, Mar. 2004.

\bibitem{Bonald04}
T.~Bonald, ``{A score-based opportunistic scheduler for fading radio
  channels},'' \emph{Proc. European Wireless}, 2004.

\bibitem{Patil09}
S.~Patil and G.~de~Veciana, ``{Measurement-based opportunistic scheduling for
  heterogeneous wireless systems},'' \emph{{IEEE} Trans. Commun.}, vol.~57,
  no.~9, pp. 2745--2753, Sep. 2009.

\bibitem{Kountouris08}
M.~Kountouris, T.~Salzer, and D.~Gesbert, ``{Scheduling for multiuser MIMO
  downlink channels with ranking-based feedback},'' \emph{EURASIP Journal on
  Advances in Signal Processing}, vol. 2008, no.~5, pp. 1--14, Jan. 2008.

\bibitem{Bang11}
H.~J. Bang and P.~Orten, ``{Scheduling and feedback reduction in cellular
  networks with coordination clusters},'' \emph{in Proc. {IEEE} Wireless
  Communications and Networking Conference (WCNC)}, pp. 1864--1868, Mar. 2011.

\bibitem{Porat09}
U.~Ben-Porat, A.~Bremler-Barr, and H.~Levy, ``{On the exploitation of CDF based
  wireless scheduling},'' \emph{in Proc. {IEEE} International Conference on
  Computer Communications (INFOCOM)}, pp. 2821--2825, Apr. 2009.

\bibitem{Patil07}
S.~Patil and G.~de~Veciana, ``{Managing resources and quality of service in
  heterogeneous wireless systems exploiting opportunism},'' \emph{{IEEE/ACM}
  Trans. Netw.}, vol.~15, no.~5, pp. 1046--1058, Oct. 2007.

\bibitem{Miao12}
G.~Miao, Y.~G. Li, and A.~Swami, ``{Channel-aware distributed medium access
  control},'' \emph{{IEEE/ACM} Trans. Netw.}, vol.~20, no.~4, pp. 1290--1303,
  Aug. 2012.

\bibitem{Wang06}
J.~Wang, H.~Zhai, Y.~Fang, J.~M. Shea, and D.~Wu, ``{OMAR: utilizing multiuser
  diversity in wireless Ad hoc networks},'' \emph{{IEEE} Trans. Mobile
  Comput.}, vol.~5, no.~12, pp. 1764--1779, Dec. 2006.

\bibitem{Gesbert04}
D.~Gesbert and M.-S. Alouini, ``{How much feedback is multi-user diversity
  really worth?}'' \emph{in Proc. {IEEE} Vehicular Technology Conference
  (VTC)}, pp. 234--238, Jun. 2004.

\bibitem{Kim07}
H.~Kim and Y.~Han, ``{An opportunistic channel quality feedback scheme for
  proportional fair scheduling},'' \emph{{IEEE} Commun. Lett.}, vol.~11, no.~6,
  pp. 510--503, Jun. 2007.

\bibitem{Jung11}
B.~C. Jung and W.-Y. Shin, ``{Opportunistic interference alignment for
  interference-limited cellular uplink},'' \emph{{IEEE} Commun. Lett.},
  vol.~15, no.~2, pp. 148--150, Feb. 2011.

\bibitem{Jung12}
B.~C. Jung, D.~Park, and W.-Y. Shin, ``{Opportunistic interference mitigation
  achieves optimal degrees-of-freedom in wireless multi-cell uplink
  networks},'' \emph{{IEEE} Trans. Commun.}, vol.~60, no.~7, pp. 1935--1944,
  Jul. 2012.

\bibitem{Ban09}
T.~W. Ban, W.~Choi, B.~C. Jung, and D.~K. Sung, ``{Multi-user diversity in a
  spectrum sharing system},'' \emph{{IEEE} Trans. Wireless Commun.}, vol.~8,
  no.~1, pp. 102--106, Jan. 2009.

\bibitem{Tajer10}
A.~Tajer and X.~Wang, ``{Multiuser diversity gain in cognitive networks},''
  \emph{{IEEE/ACM} Trans. Netw.}, vol.~18, no.~6, pp. 1766--1779, Dec. 2010.

\end{thebibliography}

\end{document}